\documentclass{siamltex1213}
\usepackage{a4, graphicx,geometry,amsmath,rotating,amssymb,epsfig,epstopdf,algpseudocode,accents,xcolor,float}
\usepackage{url}
\usepackage{graphicx}
\usepackage{amssymb}
\usepackage{epstopdf}
\usepackage{titlesec}
\usepackage[toc,page]{appendix}
\def\tallqed{\smash{\scalebox{0.65}[1.1]{$\Box$}}}

\newcommand{\inte}{\accentset{\circ}}

\title{Neighbourhoods of Phylogenetic Trees: Exact and Asymptotic Counts}
\author{J. V. de Jong,\footnotemark[2]\ 
\and J. C. McLeod\footnotemark[2]\ 
\and M. Steel\footnotemark[2]}
\begin{document}

\maketitle
\renewcommand{\thefootnote}{\fnsymbol{footnote}}
\footnotetext[2]{School of Mathematics and Statistics, University of Canterbury, Christchurch, New Zealand (jamiev.dejong@gmail.com, jeanette.mcleod@canterbury.ac.nz, mike.steel@canterbury.ac.nz). }

\renewcommand{\thefootnote}{\arabic{footnote}}

\pagestyle{myheadings}
\thispagestyle{plain}
\markboth{J. V. DE JONG, J. C. MCLEOD AND M. STEEL}{NEIGHBOURHOODS OF PHYLOGENETIC TREES}
\begin{abstract}A central theme in phylogenetics is the reconstruction and analysis of evolutionary trees from a given set of data. To determine the optimal search methods for reconstructing trees, it is crucial to understand the size and structure of the neighbourhoods of trees under tree rearrangement operations. The diameter and size of the immediate neighbourhood of a tree has been well-studied, however little is known about the number of trees at distance two, three or (more generally) $k$ from a given tree. In this paper we provide a number of exact and asymptotic results concerning these quantities, and identify some key aspects of tree shape that play a role in determining these quantities. We obtain several new results for two of the main tree rearrangement operations - Nearest Neighbour Interchange and Subtree Prune and Regraft -- as well as for the Robinson--Foulds metric on trees.
\end{abstract}



\section{Introduction}

Phylogenetics is the study of evolutionary relationships between species. These relationships are represented as phylogenetic trees, where the leaves correspond to extant species and the interior vertices correspond to ancestral species. A branch between two species in a tree indicates an evolutionary relationship between them~\cite{Book, OtherBook}. Central to phylogenetics is the problem of finding the optimal tree to fit a given data set, with the aim of determining the evolutionary history of the species being studied. However the number of possible phylogenetic trees grows rapidly with the number of leaves, so for data sets with a large number of leaves, the optimal tree is commonly found by searching the set of phylogenetic trees (tree space) via tree rearrangement operations~\cite{one, two}. Tree rearrangement operations are also used to compare phylogenetic trees by looking at the distance (smallest number of tree rearrangement operations) between the trees. These could be trees obtained from the same data set using different search methods, or from different data sets on the same set of species~\cite{distance, DasGupta}.\\

In order to effectively search tree space using tree rearrangement operations, it is crucial to understand the size and structure of the neighbourhood of (i.e. the set of trees obtained from) a phylogenetic tree under these operations. In this paper, we investigate the size of the neighbourhoods of trees arising from two commonly used tree rearrangement operations: Nearest Neighbour Interchange (NNI) and Subtree Prune and Regraft (SPR), as well as the Robinson--Foulds (RF) distance. Fig.~\ref{Intro} shows examples of the RF, NNI and SPR distances between trees. Expressions for the number of trees at distance one or two from a given tree under RF, distance one, two or three under NNI, and distance one under SPR and Tree Bisection and Reconnection (TBR) are already known~\cite{Dist, Robinson, Steel, TBR}. We provide new asymptotic expressions for the number of trees at distance $k$ from a given tree under NNI and the RF distance. We also show that unlike NNI and RF, the number of trees at distance two from a given tree under SPR is dependent on the shape of the tree, and cannot be expressed solely in terms of the number of leaves and cherries of the tree. \\

\begin{figure}[htb]
\centering
\epsfig{file = 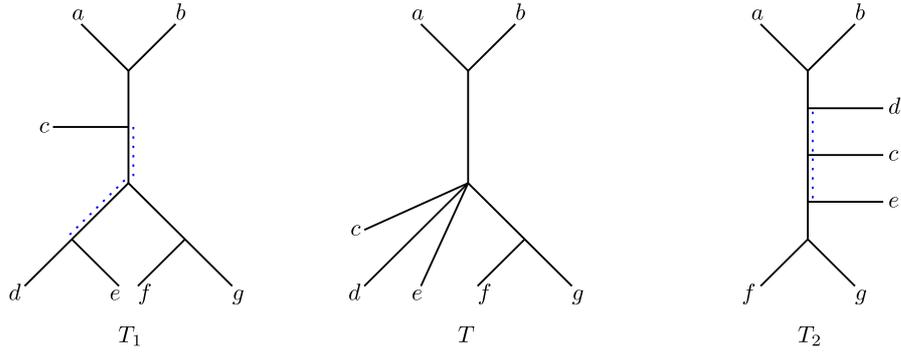, width = 0.8\linewidth}
\caption{Here, $T_1$ and $T_2$ are unrooted binary phylogenetic trees with seven leaves. They are (i) distance two apart under the RF metric, (ii) distance two apart under the NNI metric, and (iii) distance one apart under the SPR metric. Tree $T$ is obtained from $T_1$ or $T_2$ by contracting the two internal edges indicated by dotted lines.}
\label{Intro}
\end{figure}

The literature on the structure of tree neighbourhoods and tree space includes results regarding the smallest number of NNI operations required to reach every tree in the set~\cite{Hamilton, SPRWalks}, and the characterisation of the splits appearing in trees within a certain distance of a given tree under various distance measures including RF, NNI, SPR, and TBR~\cite{Bryant}. Here, we provide asymptotic results for the number of binary (fully resolved) trees that are a specified (small) distance from a given binary tree under the RF metric. Recently, Allen and Steel~\cite{Steel} established the asymptotics at the other end of the distribution. They showed that the proportion of binary trees that are at nearly maximal distance from each other follows a Poisson distribution whose mean depends on the proportion of leaves of the given tree that lie in a cherry (a path of length two where both endpoints are leaves of the tree). Using the expressions for the sizes of the first and second neighbourhoods, we provide an exact count for the number of pairs of binary phylogenetic trees with $n$ leaves that share a first neighbour under NNI and RF. 

\section{Definitions}

A {\it graph} $G$ is an ordered pair $(V(G),E(G))$ consisting of a vertex set $V(G)$ and an edge set~$E(G)$. For any vertices $x,y\in V(G)$, $x$ and $y$ are {\it adjacent} if there is an edge $e\in E(G)$ such that~$e=\{x,y\}$. We call $x$ and $y$ the {\it endpoints} of~$e$, and vertex $x$ and edge $e$ are said to be {\it incident}. Two distinct edges $e,f\in E(G)$ are {\it adjacent} if they have an endpoint in common. Edges $e,f\in E(G)$ are {\it parallel edges} if they have the same endpoints. An edge $f=\{x,x\}$ where $x\in V(G)$ is called a {\it loop}. A graph is {\it simple} if it has no loops or parallel edges. All of the graphs referred to in this paper are simple. \\

The {\it degree} of a vertex $v\in V(G)$ is the number of vertices in $V(G)$ that are adjacent to~$v$, and is denoted~$deg(v)$. The Handshaking Lemma is a well-known result stating that for a graph~$G$, $\sum_{v\in V(G)}deg(v)=2|E(G)|$ (see~\cite{graph} for more detail). \\

A graph $H$ is a {\it subgraph} of $G$ if $V(H)\subseteq V(G)$ and $E(H)\subseteq E(G)$. If $V(H)\subset V(G)$ or $E(H)\subset E(G)$, then $H$ is a {\it proper subgraph} of~$G$. A {\it path} $P$ in $G$ is a subgraph of $G$ which consists of a sequence of distinct vertices $v_0,v_1,...,v_k$ such that for all $i\in \{0,1,...,k-1\}$, $v_i$ and $v_{i+1}$ are adjacent in~$P$. We call this a {\it path of length $k$}. We may also refer to $P$ as a $(v_0{-}v_k)$-path or an $(e{-}f)$-path where $e=\{v_0,v_1\}$ and $f=\{v_{k-1},v_k\}$. Note that paths are regarded as undirected, so a $(v_0-v_k)$-path in $G$ is considered identical to a $(v_k-v_0)$-path in $G$. A {\it cycle} is a path in which the first and last vertices are the same (i.e.~$v_0=v_k$). The subgraph of a graph $G$ {\it induced} by the vertex set $V\subseteq V(G)$ is the subgraph with vertex set~$V$ and edge set $E\subseteq E(G)$, where $E$ consists of all the edges of $G$ that have both endpoints in~$V$. \\

Two vertices $x,y\in V(G)$ are {\it connected} if there is an $(x{-}y)$-path in~$G$. A graph $G$ is {\it connected} if all pairs of vertices $x,y\in V(G)$ are connected. A {\it component} of $G$ is a maximal connected subgraph of~$G$.\\

The {\it distance} between two vertices $x,y\in V(G)$, denoted $d_G(x,y)$, is the length of the shortest $(x{-}y)$-path in~$G$. We define the {\it distance} between two vertex sets, $U=\{u_{1},u_{2},...\}$ and \text{$V=\{v_{1},v_{2},...\}$} to be $d_G(U,V)$, where
$$d_G(U,V)=\min\{d_G(u_i,v_j): 1\leq i\leq |U|, 1\leq j\leq |V|\}.$$ 
The diameter $M$ of $G$ is given by
$$M=\max\{d_G(v_i,v_j): v_i,v_j\in E(G)\}.$$

Two graphs $G$ and $G'$ are {\it isomorphic} if there is a bijection $\sigma:V(G)\rightarrow V(G')$ such that for all pairs of vertices $x,y\in V(G)$, $x$ and $y$ are adjacent in $G$ if and only if $\sigma(x)$ and $\sigma(y)$ are adjacent in~$G'$. \\

\subsection{Trees}

A {\it tree} $T$ is a connected graph containing no cycles. A {\it forest} is a graph whose components are trees. A tree is {\it rooted} if it has a distinguished root vertex; otherwise, it is {\it unrooted}. A {\it leaf} of a tree $T$ is a vertex of $T$ that has degree one. The {\it leaf set} $\mathcal{L}(T)\subseteq V(T)$ of a tree $T$ is the set of all leaves in~$T$. Vertices of $T$ that are not leaves, are called {\it internal vertices}. If an edge of $T$ is incident to a leaf, we call it a {\it pendant edge} of~$T$; otherwise, it is an {\it internal edge} of~$T$. \\

A {\it binary tree} is a tree in which all internal vertices have degree three. A {\it binary phylogenetic tree} $T$ is a binary tree with a bijection $\phi: X\rightarrow \mathcal{L}(T)$ where $X$ is a set of $n$ labels (see Fig.~\ref{Intro}). Let $UB(n)$ be the set of all unrooted binary phylogenetic trees with $n$ leaves. For trees $T_1,T_2\in UB(n)$, we say that $T_1$ and $T_2$ are {\it equal} ($T_1=T_2$) if they are isomorphic by a map that preserves the leaf labelling. In this paper we will use the term `tree' to refer to an unrooted binary phylogenetic tree unless otherwise stated.\\

A {\it cherry} in a tree $T$ is a path of length two in which both endpoints are leaves of~$T$. Let $UB(n,c)$ be the set of all unrooted binary phylogenetic trees with $n$ leaves and $c$ cherries. For example, in Fig.~\ref{Intro}, $T_1 \in UB(7,3)$ and $T_2 \in UB(7,2)$, while $T$ is not a binary tree. \\

\subsection{Subtrees}

A {\it subtree} of a graph $G$ is a subgraph of $G$ that is a tree. All connected subgraphs of a tree $T$ are subtrees. The {\it distance} in $T$ between a subtree $T'$ of $T$ and a set of vertices $V\subseteq V(T)$ is~$d_T(V(T'),V)$, but we will simply write this as~$d_T(T',V)$. Throughout this paper we assume that all subtrees are proper subtrees, and have the property that if $T'$ is a subtree of $T\in UB(n)$, then $\mathcal{L}(T')\subseteq\mathcal{L}(T)$. This ensures that $T'$ has at least one vertex of degree two. If $T'$ has exactly one vertex of degree two then it is a {\it pendant subtree}; otherwise, it is an {\it internal subtree}.  Unless otherwise specified, all subtrees in this paper are maximal, pendant subtrees. We may refer to the vertex $v$ of degree two in a pendant subtree $T'$ as the {\it root} of $T'$.  \\

An edge $e$ in a tree $T$ is {\it incident} to a subtree $T'$ of $T$ if $e$ is incident to a vertex of degree two in~$T$. The {\it intersection} $T_1\cap T_2$ of two subtrees $T_1$ and $T_2$ of a tree $T$, is a subtree of $T$ in which $V(T_1\cap T_2)=V(T_1)\cap V(T_2)$ and $E(T_1\cap T_2)=E(T_1)\cap E(T_2)$. \\

A tree $T$ is a {\it caterpillar} if the subtree induced by the internal vertices of $T$ is a path. A {\it balanced tree} is a tree in which all leaves are equidistant from a single vertex or edge. Fig.~\ref{Catbal} shows a caterpillar and two balanced trees, all of which are binary.\\

\begin{figure}[H]
\centering
\epsfig{file = 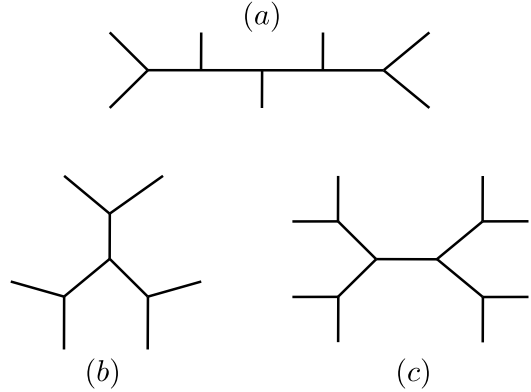, width = 0.5\linewidth}
\caption{Examples of (a) a caterpillar and (b), (c) balanced trees.}
\label{Catbal}
\end{figure}

We define $P_k(T)$ to be the number of distinct paths of length $k$ in~$T$. An {\it internal path} $P$ of a tree $T$ is a path in which all vertices of $P$ are internal vertices of~$T$. We use $p_k(T)$ to denote the number of distinct internal paths of length $k$ in~$T$. We have the following lemma relating the paths and internal paths of a tree. The proof is straightforward.\\

\begin{lemma}
Let $T\in UB(n)$ where $n\geq 4$. Then for $k\geq 3$, $P_k(T)=4p_{k-2}(T)$. 
\label{ps}
\end{lemma}

\subsection{Edge and Vertex Operations}

Given a tree~$T\in UB(n)$, if we {\it delete} an edge $e\in E(T)$, we obtain the forest $T\setminus e$ where $V(T\setminus e)=V(T)$ and $E(T\setminus e)=E(T)-\{e\}$. We {\it contract} an edge $e=\{x,y\}$ of $T$ to obtain a new non-binary tree, denoted~$T/e$, by deleting $e$ and combining $x$ and $y$ into a single vertex~$w$, such that all vertices adjacent to $x$ or $y$ in $T$ are adjacent to $w$ in~$T/e$. Fig.~\ref{Intro} shows the tree $T$ resulting from the contraction of two internal edges of a tree~$T_1$.\\

Let $T_1$ be a tree with edge~$e=\{x,y\}$. We {\it subdivide}~$e$ by deleting $e$ and inserting a vertex $u$ and edges $e_1=\{x,u\}$ and $e_2=\{u,y\}$ to obtain a non-binary tree~$T_1'$. Given a non-binary tree $T_2$, we {\it suppress} a vertex $u$ of degree two, by deleting $u$ and its incident edges $e_1=\{x,u\}$ and $e_2=\{u,y\}$, and inserting a single edge $e'=\{x,y\}$ to obtain a tree~$T_2'$. Edge subdivision and vertex suppression are inverse operations.\\

In this paper, when we perform any of the above operations, we assume that all edge and vertex labels in the original tree are preserved by the operation, except those explicitly deleted or inserted.  

\subsection{Splits}
A {\it partition} of a set $X$ is a set of disjoint, non-empty subsets \text{$\{X_1, X_2,..., X_m\}$}, $m\geq 1$, such that $X=\cup_{k=1}^m X_k$. A partition of $X$ is a {\it bipartition} if~$m=2$. Consider a tree~$T\in UB(n)$. A bipartition $\{L_1,L_2\}$ of $\mathcal{L}(T)$ is a {\it split} of $T$ if there exists an edge $e\in E(T)$ such that $T\setminus e$ has components $T_1$ and $T_2$ with $\mathcal{L}(T_1)=L_1$ and $\mathcal{L}(T_2)=L_2$. We define $S(T, e)=\{L_1,L_2\}$ as the split of $T$ {\it associated with}~$e$. A split $S(T, e)$ is {\it trivial} if $e$ is a pendant edge of~$T$. We define $\Sigma(T)=\{S(T, e)\text{: where $e$ is an internal edge of $T$} \}$ as the set of all non-trivial splits of~$T$. Two trees $T_1$ and $T_2$ are {\it equal} if and only if $\Sigma(T_1)=\Sigma(T_2)$~\cite{Buneman}.\\

The following lemma gives two well known expressions, one for the number of internal edges in a binary tree, and one for $|UB(n)|$ (see~\cite{Book} for more detail).\\

\begin{lemma}\mbox{}
\begin{enumerate}
\item[(i)] Let $T\in UB(n)$, $n\geq 3$. Then $T$ has $n-3$ internal edges.
\item[(ii)] For all $n\in\mathbb{Z}^+$, $n\geq 3$, we have $|UB(n)|=\frac{(2n-4)!}{(n-2)!2^{n-2}}.$
\end{enumerate}
\label{intedges}
\end{lemma}

\subsection{Neighbourhoods}
In this paper we consider three metrics: Robinson--Foulds (RF), Nearest Neighbour Interchange (NNI), and Subtree Prune and Regraft (SPR), defined in Section~3, Section~4, and Section~6 respectively.\\

Given one of these three metrics~$\delta_{\theta}$, $\theta\in\{\text{RF, NNI, SPR}\}$, on $UB(n)$, the {\it $k^{th}$ neighbourhood} of a tree~$T$, denoted $N_{\theta}^k(T)$, is given by
$$N_{\theta}^k(T)=\{T'\in UB(n):\delta_{\theta}(T,T')=k\}.$$
A tree $T'\in N_{\theta}^k(T)$ is called a {\it $k^{th}$ neighbour} of~$T$. \\

\section{Robinson--Foulds Metric}
Given two trees $T_1, T_2\in UB(n)$, the {\it Robinson--Foulds (RF) distance} between $T_1$ and $T_2$ is defined by
$$\delta_{RF}(T_1,T_2)=\frac{1}{2}|\Sigma(T_1)-\Sigma(T_2)|+\frac{1}{2}|\Sigma(T_2)-\Sigma(T_1)|.$$
Alternatively, the RF distance between $T_1$ and $T_2$ can be seen as the minimum $m$ for which there exist $E_1\subseteq E(T_1)$ and $E_2\subseteq E(T_2)$ where $|E_1|=|E_2|=m$, such that
$T_1/E_1=T_2/E_2.$ This is illustrated in Fig.~\ref{Intro}, where $\delta_{RF}(T_1,T_2)=2$.\\

The {\it $k^{th}$ RF neighbourhood} of a tree $T\in UB(n)$ is the set of trees in $UB(n)$ that are exactly RF distance $k$ from~$T$. In terms of edge contraction, this neighbourhood consists of all trees $T'\in UB(n)$ such that the minimum $j$ for which we could contract $j$ edges of $T$ and $j$ edges of $T'$ and obtain the same (non-binary) tree, is~$k$.\\

The RF distance was originally introduced by Bourque~\cite{RFDef} and was generalised by Robinson and Foulds~\cite{RF}. Unlike the metrics induced by NNI and SPR that we will see in later sections, the RF distance between two trees is computationally easy to calculate. (Day~\cite{Day1985} provided a linear-time algorithm.) In this section we consider the first, second and $k^{th}$ RF neighbourhoods of an unrooted binary phylogenetic tree. Let $T\in UB(n,c)$ where $n\geq 3$ (recall that $UB(n,c)$ is the set of unrooted binary phylogenetic trees with $n$ leaves and $c$ cherries). Then 
\begin{enumerate}
\item $|N_{RF}(T)|=2(n-3)$, and
\item $|N_{RF}^2(T)|=2n^2-8n+6c-12$. 
\end{enumerate}
This expression for the size of the first RF neighbourhood is commonly known, and as we will see later, is the same as the size of the first NNI neighbourhood, found by Robinson \cite{Robinson}. The expression for the size of the second RF neighbourhood appears in Section 4.2 of \cite{Dist}.\\

Much of the literature on the RF distance has focused on calculating the RF distance between two trees, and on the distribution of the distances between trees. Bryant and Steel~\cite{Dist} gave a polynomial-time algorithm for finding the distribution of trees around a given tree~$T$, and showed that this distribution can be approximated by a Poisson distribution determined by the proportion of leaves of $T$ that are in cherries. Hendy {\it et al.}~\cite{HLP} used generating function techniques to calculate the probability that two trees, selected uniformly at random from $UB(n)$, are RF distance $m$ from each other. \\

While the sizes of the first and second RF neighbourhoods are known, the sizes of higher neighbourhoods are not known in general. Although $N^2_{RF}$ depends on the shape of $T$ (via $c$), for $k=1,2$ we can write $N^k_{RF} = \frac{2^kn^k}{k!}(1+O(n^{-1}))$. Our main result in this section (Theorem 3.1) provides a generalisation of this asymptotic equality to all values of $k\geq 1$.\\

\begin{theorem}
Let $T\in UB(n)$ ($n\geq 4$). For each fixed $k\in\mathbb{Z}^+$, 
\begin{equation}N^k_{RF}(T)=\frac{2^kn^k}{k!}\left(1+C_{T,k}n^{-1}+O(n^{-2})\right),\label{rfeqn}\end{equation}
where 
$$-\frac{5k^2+7k}{4} \leq C_{T,k} \leq 4k^2-7k.$$
\label{mainRF}
\end{theorem}

The proof of Theorem~\ref{mainRF} comprises two steps. First, we determine the number of binary phylogenetic trees whose splits differ from $\Sigma(T)$ by exactly the $k$ splits associated with a given subset of $k$ internal edges of~$T$. We then determine the number of subsets of $k$ internal edges in~$T$. We consider three cases:
\begin{enumerate}
\item The $k$ edges are pairwise non-adjacent.
\item Exactly two of the $k$ edges are adjacent.
\item More than two of the $k$ edges are adjacent.\\
\end{enumerate}
The term of order $n^k$ in Equation~(\ref{rfeqn}) is completely determined by Case 1 above, while the term of order $n^{k-1}$ is determined by Cases 1 and 2. We show that all other possibilities for the $k$ edges, (covered by Case 3) only contribute to terms of order $n^{k-2}$ or lower.\\

\textbf{Neighbours with Different Splits over $k$ Given Edges}\\

Let $\Sigma_k$ be a given set of $k$ splits of $T\in UB(n)$ ($k\geq 1$). We define
$$\Delta(T,\Sigma_k)=|\{T'\in UB(n): (\Sigma(T)-\Sigma_k)\subset \Sigma(T') \}|,$$
as the number of trees containing the splits $\Sigma(T)-\Sigma_k$; and
$$\inte{\Delta}(T,\Sigma_k)=|\{T'\in UB(n): (\Sigma(T)\cap\Sigma(T'))=\Sigma(T)-\Sigma_k\}|,$$
as the number of trees containing the splits $\Sigma(T)-\Sigma_k$, and no other splits of~$T$. \\

In Lemma~\ref{prod1} we obtain an expression for $\Delta(T,\Sigma_k)$ and show that once $T$ and $\Sigma_k$ are specified, $\inte{\Delta}(T,\Sigma_k)$ is independent of~$n$. \\

\newpage
\begin{lemma} 
Let $T\in UB(n)$ ($n\geq 4$), let $e_1$, ..., $e_k$ ($1\leq k\leq n-3$) be distinct internal edges of~$T$, and let $\Sigma_k$ be the set of $k$ splits of $T$ associated with these edges. Define $F$ to be the subgraph of $T$ consisting of the edges $e_1$,...,~$e_k$. Then
\begin{enumerate}
\item[(i)] $$\Delta(T,\Sigma_k)=\prod_{m=1}^k \left(\frac{(2m+2)!}{(m+1)!2^{m+1}}\right)^{c_m},$$

where $c_m$ is the number of components of $F$ with exactly $m$ edges.
\item[(ii)] Let $T'\in UB(s)$ ($s\geq k+3$), and let $F'$ be the subgraph of $T'$ consisting of distinct internal edges $e_1'$, ..., $e_k'$ of $T'$. Let $\Sigma_k'$ be the set of $k$ splits of $T'$ associated with these edges. If $F'$ is isomorphic to $F$, then 
$$\inte{\Delta}(T',\Sigma_k')=\inte{\Delta}(T,\Sigma_k).$$
In other words, the number of trees containing the splits $\Sigma(T)-\Sigma_k$ and no other splits of $T$, is not dependent on $n$.
\end{enumerate}
\label{prod1} 
\end{lemma}

\begin{proof} \mbox{}
\begin{enumerate}
\item[(i)] Let $C_1,..., C_\ell$ be the components of $F$. Given a component $C_i$ with $m$ edges, let $A_i$ be the subtree of $T$ consisting of the corresponding $m$ edges of $F$ in $T$ and their adjacent edges (note that $A_i$ may be an internal subtree). Then $A_i$ has $m+3$ leaves. We want to find $\Delta(A_i,\Sigma_i)$, where $\Sigma_i$ is the set of splits associated with the internal edges of~$A_i$. (Note that this is the same as the number of trees that are at most RF distance $m$ from~$A_i$.) Clearly $\Delta(A_i,\Sigma_i)=|UB(|\mathcal{L}(A_i)|)|=|UB(m+3)|$, as it is the number of trees in $UB(m+3)$ that have at least zero splits in common with~$A_i$. By Lemma~\ref{intedges},
$$|UB(m+3)|=\frac{(2(m+3)-4)!}{((m+3)-2)!2^{(m+3)-2}}=\frac{(2m+2)!}{(m+1)!2^{m+1}}.$$
We can apply this principle to each component of~$F$. The results for each component are independent of those for the other components of~$F$. Therefore, we can take the product to obtain
$$\Delta(T,\Sigma_k)=\prod_{i=1}^\ell \Delta(A_i,\Sigma_i)=\prod_{m=1}^k \left(\frac{(2m+2)!}{(m+1)!2^{m+1}}\right)^{c_m}.$$

\item[(ii)] This is similar to~$(i)$, except that we now restrict our attention to $\inte{\Delta}(T,\Sigma_k)$, that is, those trees in $\Delta(T,\Sigma_k)$ that do not contain any of the splits in~$\Sigma_k$. Similarly to $(i)$, we have
\begin{equation}\inte{\Delta}(T,\Sigma_k)=\prod_{i=1}^\ell \inte{\Delta}(A_i,\Sigma_i).\label{delta}\end{equation}
Note that for each subtree $A_i$, some of the trees counted by $\Delta(A_i,\Sigma_i)$ have splits in common with~$A_i$, and hence are not counted by~$\inte{\Delta}(A_i,\Sigma_i)$. Clearly, $\inte{\Delta}(A_i,\Sigma_i)$ is dependent on the shape and size of~$A_i$, which itself depends on the choice of the $k$ edges of $T$ and not on the shape or number of leaves of~$T$. Therefore, since $F'=F$, we have
$$\inte{\Delta}(T',\Sigma_k')=\prod_{i=1}^\ell \inte{\Delta}(A_i,\Sigma_i)=\inte{\Delta}(T,\Sigma_k).$$\\
\end{enumerate}
\end{proof}

We now consider expressions for $\Delta(T,\Sigma_k')$ and $\inte{\Delta}(T,\Sigma_k')$ where $\Sigma_k'$ is the set of splits associated with $k$ distinct, pairwise non-adjacent internal edges of $T$.\\

\begin{lemma}
Let $T\in UB(n)$ ($n\geq 4$) and let $\Sigma_k'$ ($1\leq k\leq n-3$) be the set of splits associated with distinct, pairwise non-adjacent internal edges $e_1,...,e_k$ of~$T$. Then
\begin{enumerate}
\item[(i)] $\Delta(T,\Sigma_k')=3^k$, and
\item[(ii)] $\inte{\Delta}(T,\Sigma_k')=2^k$.\\
\end{enumerate}
\label{3k2k}
\end{lemma}
\begin{proof}\mbox{}
\begin{enumerate}
\item[(i)] This follows directly from Lemma~\ref{prod1}.\\

\item[(ii)] For some $i$, $1\leq i\leq k$, let $A_i$ be the subtree of $T$ consisting of edge $e_i$ and its adjacent edges in~$T$ (note that $A_i$ may be an internal subtree). Then $A_i$ has four leaves. Note that $\Delta(A_i,S(A_i,e_i))=|UB(4)|=3$. However, one of these three trees is~$A_i$. The remaining two trees each have a single internal edge, and the split associated with this edge is not $S(A_i,e_i)$. Hence $\inte{\Delta}(A_i,S(A_i,e_i))=2$, and by Equation~(\ref{delta}), $\inte{\Delta}(T,\Sigma_k')=2^k$.
\end{enumerate}
\end{proof}

Now that we have investigated the case where the $k$ internal edges are pairwise non-adjacent, we consider an adjacent pair of internal edges. \\

\begin{lemma}
Let $T\in UB(n)$ ($n\geq 5$), and let $\Sigma_2$ be the set of splits associated with two adjacent internal edges of $T$. Then
$\inte{\Delta}(T,\Sigma_2)=10.$\\
\label{new34}
\end{lemma}
\begin{proof}
Let $\Sigma_2=\{S(T,e_1),S(T,e_2)\}$, where $e_1$ and $e_2$ are adjacent internal edges of~$T$. The set of trees counted by $\Delta(T,\Sigma_2)$ includes trees which have one or more of the splits in $\Sigma_2$ in common with~$T$. Therefore, to obtain $\inte{\Delta}(T,\Sigma_2)$, we subtract from $\Delta(T,\Sigma_2)$ the number of trees in $UB(n)$ that have exactly one split different to $T$ associated with either $e_1$ or~$e_2$, or the same splits as~$T$. Hence 
$$\inte{\Delta}(T,\Sigma_2)=\Delta(T,\Sigma_2)-\inte{\Delta}(T,S(T,e_1))-\inte{\Delta}(T,S(T,e_2))-1.$$
Therefore, by Lemma~\ref{prod1}, if $e_1$ and $e_2$ are adjacent, then $\inte{\Delta}(T,\Sigma_k)=15-5=10.$\\
\end{proof}

\textbf{The Number of Subsets of $k$ Internal Edges}\\

\begin{lemma}
Let $T\in UB(n)$ ($n\geq 4$). Then
\begin{enumerate}
\item[(i)] The number of sets of $k$ distinct, pairwise non-adjacent internal edges $e_1$,..., $e_k$ in $T$ ($1\leq k\leq n-3$), denoted $A_{T,k}$, satisfies
$$\frac{1}{k!}n^k-\frac{k(5k+1)}{2k!}n^{k-1}+O(n^{k-2})\leq A_{T,k} \leq \frac{1}{k!}n^k-\frac{k(k+2)}{k!}n^{k-1}+O(n^{k-2}).$$
\item[(ii)] The number of sets of $k$ distinct internal edges $e_1$,..., $e_k$ in $T$ ($2\leq k\leq n-3$) where exactly two edges are adjacent, denoted $B_{T,k}$, satisfies
$$\frac{1}{2(k-2)!}n^{k-1}+O(n^{k-2})\leq B_{T,k} \leq \frac{2}{(k-2)!}n^{k-1}+O(n^{k-2}).$$
\item[(iii)] The number of sets of $k$ distinct internal edges $e_1$,..., $e_k$ ($3\leq k\leq n-3$) in $T$ where more than two edges are adjacent is~$O(n^{k-2})$.\\
\end{enumerate}
\label{counting}
\end{lemma}
\newpage
\begin{proof}
\mbox{}
\begin{enumerate}
\item[(i)] We calculate the bounds by considering the best and worst case scenarios for the choice of each edge. There are $n-3$ choices for the first edge~$e_1$. There are at most $(n-3)-2$ choices for $e_2$ (this can occur when $e_1$ has exactly one adjacent internal edge in~$T$). There are then at most $(n-3)-4$ choices for $e_3$ (this can occur when $e_1$ and $e_2$ each have exactly one adjacent internal edge in~$T$), and so on. Therefore 
\begin{align*}
A_{T,k}&\leq \frac{1}{k!}(n-3)(n-3-2)(n-3-2(2))\cdots(n-3-2(k-1))\\
&=\frac{1}{k!}n^k-\frac{1}{k!}n^{k-1}\sum_{i=0}^{k-1}(3+2i)+O(n^{k-2})\\
&=\frac{1}{k!}n^k-\frac{k(k+2)}{k!}n^{k-1}+O(n^{k-2}).
\end{align*}

On the other hand, there are at least $(n-3)-5$ choices for $e_2$ (this can occur when $e_1$ has four adjacent internal edges in~$T$). There are then at least $(n-3)-10$ choices for $e_3$ (this can occur when $e_1$ and $e_2$ each have four adjacent internal edges in~$T$), and so on. Therefore
\begin{align*}
A_{T,k}&\geq \frac{1}{k!}(n-3)(n-3-5)(n-3-5(2))\cdots(n-3-5(k-1))\\
&=\frac{1}{k!}n^k-\frac{1}{k!}n^{k-1}\sum_{i=0}^{k-1}(3+5i)+O(n^{k-2})\\
&=\frac{1}{k!}n^k-\frac{k(5k+1)}{2k!}n^{k-1}+O(n^{k-2}).
\end{align*}

\item[(ii)] We will prove this in the same way as (i), assuming without loss of generality that $e_1$ and $e_2$ are the adjacent pair of edges. There are $n-3$ choices for~$e_1$. There are at most four choices for $e_2$ (this can occur if $e_1$ has four adjacent internal edges in~$T$). For $e_3$, there are at most $(n-3)-3$ choices (this can occur if $e_1$ and $e_2$ each have two adjacent pendant edges in~$T$). The remaining edges follow in the same way as in (1). Therefore 
\begin{align*}
B_{T,k}&\leq \frac{4}{2(k-2)!}(n-3)(n-6)(n-6-2(1))...(n-6-2(k-3))\\
&=\frac{2}{(k-2)!}n^{k-1}+O(n^{k-2}).
\end{align*}

On the other hand, there is at least one choice for $e_2$ (this can occur if $e_1$ has exactly one adjacent internal edge in~$T$). For $e_3$, there are at least $(n-3)-7$ choices (this can occur if $e_1$ and $e_2$ each have no adjacent pendant edges in~$T$). The remaining edges are chosen in the same way as in (1). Hence
\begin{align*}
B_{T,k}&\geq \frac{1}{2(k-2)!}(n-3)(n-10)(n-10-5(1))...(n-10-5(k-3))\\
&=\frac{1}{2(k-2)!}n^{k-1}+O(n^{k-2}).
\end{align*}

\item[(iii)] Let $F$ be the subgraph of $T$ consisting of the edges $e_1$,...,~$e_k$. Then $F$ has $m\leq k-2$ components. Suppose we first choose $m$ internal edges of $T$ corresponding to one edge in each component of~$F$. By $(i)$, the number of such choices is $O(n^{m})$, as each of these edges will contribute a linear factor to the total number of ways of choosing the $k$ edges. However, the remaining $k-m\geq 2$ edges must be chosen from edges that are adjacent to those already chosen. The number of these choices depends only on the number and location of the edges already chosen, and not on~$n$. Hence the number of possible sets is $O(m)$, where~$m\leq k-2$.\\
\end{enumerate}
\end{proof}

Note that in the proof of Lemma~\ref{counting}, it may not be possible to maximise (or minimise) the number of choices for each individual edge in $T$, however, this is not a problem as we only require bounds on the number of choices of the $k$ edges of~$T$.\\ 

From Lemma~\ref{prod1}, Lemma~\ref{3k2k}, and Lemma~\ref{new34}, we know the number of binary phylogenetic trees whose splits differ from those of $T\in UB(n)$ by exactly $k$ splits over a given set of $k$ edges. From Lemma~\ref{counting}, we have number of subsets of $k$ internal edges. We are now in a position to prove Theorem~\ref{mainRF}.

\subsection*{Proof of Theorem~\ref{mainRF}}
We break down the calculation of the size of the $k^{th}$ RF neighbourhood of $T$ into two steps. We consider the number of trees whose splits differ from those of $T$ by exactly the $k$ splits corresponding to a given set of $k$ distinct internal edges of~$T$. We then consider the number of ways these $k$ edges can be chosen in~$T$. By Lemma~\ref{prod1}, given $T$ and a set of $k$ distinct internal edges of $T$ with associated split set~$\Sigma_k$, the number of trees with the splits $\Sigma(T)-\Sigma_k$ and none of the splits in~$\Sigma_k$ ($\inte{\Delta}(T,\Sigma_k)$), is independent of~$n$. Hence, only the number of ways of choosing the $k$ edges in $T$ is dependent on~$n$.  \\

By Lemma~\ref{counting}, when we count the number of ways of choosing $k$ distinct internal edges of~$T$, the case where the $k$ edges are pairwise non-adjacent (Case 1 from the beginning of this section) gives a term of order $n^k$ and a term of order~$n^{k-1}$. The case where exactly two of the $k$ edges are adjacent (Case 2) produces a term of order~$n^{k-1}$, but does not have a term of order~$n^{k}$. If more than two of the $k$ edges are adjacent (Case 3), then the highest order term is~$O(n^{k-2})$. \\

Now we consider the number of trees whose splits differ from those of $T$ by exactly the $k$ splits corresponding to a given set of $k$ distinct internal edges of~$T$. From the information above, the only two cases we need to consider are those where the $k$ edges are pairwise non-adjacent, or exactly two of the $k$ edges are adjacent. By Corollary~\ref{3k2k}, the case where all edges are pairwise non-adjacent produces $2^k$ $k^{th}$ RF neighbours with splits that differ from the splits of $T$ over precisely the $k$ given internal edges. In the case where exactly two edges are adjacent, the $k-1$ pairwise non-adjacent edges give $2^{k-2}$ $k^{th}$ RF neighbours, by Corollary~\ref{3k2k}. By Lemma~\ref{new34}, the adjacent pair of edges results in $10$ neighbours. Hence, in total, there are $10\cdot 2^{k-2}$ $k^{th}$ RF neighbours. Therefore, by Lemma~\ref{counting},

\begin{align*}
|N^k_{RF}(T)|&\geq \left(\frac{1}{k!}n^k-\frac{k(5k+1)}{2k!}n^{k-1}\right)2^k+10\left(\frac{1}{2(k-2)!}n^{k-1}\right)2^{k-2}+O(n^{k-2})\\
&=\frac{2^k}{k!}n^k-\frac{5k^2+7k}{4k!}2^kn^{k-1}+O(n^{k-2}).\\
\end{align*}
\begin{align*}
|N^k_{RF}(T)|&\leq \left(\frac{1}{k!}n^k-\frac{k(k+2)}{k!}n^{k-1}\right)2^k+10\left(\frac{2}{(k-2)!}n^{k-1}\right)2^{k-2}+O(n^{k-2})\\
&=\frac{2^k}{k!}n^k+\frac{4k^2-7k}{k!}2^kn^{k-1}+O(n^{k-2}).\\
\end{align*}
\tallqed

\subsection{Shared Splits}

In this subsection, we present a simple and general upper bound on the proportion of
binary trees that share at least $k$ non-trivial splits with a given tree on the same leaf set.  The relevance of this result for biology is that it shows that a `random' binary tree (selected with uniform probability) has a low probability of sharing more than a few splits with a given tree, regardless of the number of leaves (species) involved and the topology of the given tree.  For example, the probability of sharing three non-trivial splits is at most~$0.02$.\\

Let $T_0$ be a phylogenetic tree with $n$ leaves, and let $\pi_k(T_0)$ be the proportion of trees in $UB(n)$ that share at least $k$ non-trivial splits with~$T_0$ (note that $T_0$ does not have to be a binary tree). 
Thus, $\pi_k(T_0) $ is the proportion of binary phylogenetic trees $T$ for which $$d_{RF}(T, T_0) \leq \frac{1}{2}(|i_0|+n-3 -2k),$$ where $i_0$ is the number of internal edges of $T_0$. 
In general, $\pi_k(T_0)$ will depend on properties of the tree $T_0$; however, there is a universal upper bound on $\pi_k$ that applies for any choice of $T_0$ and
is independent of the number of internal edges in $T$, and even of~$n$.  This is provided by the following result.\\

\begin{theorem}
For any phylogenetic tree $T_0$ with $n$ leaves, the proportion, $\pi_k(T_0)$, of tree in $UB(n)$ that share at least $k$ non-trivial splits with $T_0$ satisfies
$$\pi_k(T_0) \leq \frac{1}{2^k k!},$$
for all $k =1, 2, \cdots, i_0$ and $\pi_k(T_0)=0$ for all~$k>i_0$, where $i_0$ is the number of internal edges of~$T_0$.\\
\label{MikeThm}
\end{theorem}

\begin{proof}
Let $N_k(T_0)$ be the number of trees in $UB(n)$ that share at least $k$ non-trivial splits with~$T_0$.
Let $\Sigma_0 = \Sigma(T_0)$, the set of non-trivial splits of~$T_0$.
We have

$$N_k(T_0) =\left |\bigcup_{\overset{\Sigma \subseteq \Sigma_0: }{\overset{|\Sigma|= k}{}}} \{T \in UB(n): \Sigma \subseteq \Sigma(T)\}\right|.$$
Therefore, by the union bound, we have
$$N_k(T_0) \leq \sum_{\overset{\Sigma \subseteq \Sigma_0: }{\overset{|\Sigma|= k}{}}}|\{T \in UB(n): \Sigma \subseteq \Sigma(T)\}|.$$
Since there are precisely $\binom{|i_0|}{k}$ terms in this sum, we obtain
\begin{equation}
\label{bo}
N_k(T_0) \leq    \binom{i_0}{k} \cdot M,
\end{equation}
where $$M = \max_{\overset{\Sigma \subseteq \Sigma_0: }{\overset{|\Sigma|= k}{}}} |\{T \in UB(n): \Sigma \subseteq \Sigma(T)\}|.$$
Now, for a subset $\Sigma$ of $\Sigma_0$ let $T_\Sigma$ be the unique non-binary phylogenetic tree that has $\Sigma$ as its set of not-trivial splits (i.e. the tree obtained from
$T_0$ by contracting each internal edge of $T_0$ that is not associated with a split  in $\Sigma$). Let $\inte{V}(T_\Sigma)$ denote the set of interior vertices of~$T_\Sigma$.
Then 

\begin{equation}
\label{count}
|\{T \in UB(n): \Sigma \subseteq \Sigma(T)\}| = \prod_{v \in V_{\rm int}(T_\Sigma)}|UB(deg(v))|.
\end{equation}
This is similar to the result in Lemma~\ref{prod1}(i), but $T_0$ is not binary. \\

Now, for each vertex $v \in \inte{V}(T_{\Sigma})$, we have ${\rm deg}(v) \geq 3$.  Moreover, when  $|\Sigma| =k$, a simple counting argument shows that
\begin{equation}
\label{counts} |\inte{V}(T_{\Sigma})|=k+1  \mbox{ and } \sum_{v \in \inte{V}(T_{\Sigma})}{\rm deg}(v) = n+2k.
\end{equation}
Leaving trees for a moment, consider the optimisation problem of maximising  $\prod_{i=1}^N b(n_i),$
subject to the constraints that  $n_1, n_2, \cdots, n_N$ are integers, each taking a value of at least $3$,  and with a sum equal to $R \geq 3N$.  
 It follows from the faster-than-exponential growth of the function $b$ that the maximum possible value is $b(R-3(N-1))$ 
(Lemma 5 of~\cite{MikesRef}). Taking $N=k+1$ and $R= n+2k$ (from Equation~(\ref{counts})), so that $R-3(N-1) = n-k$, we see from Equation~(\ref{count}) that
$|\{T \in UB(n): \Sigma \subseteq \Sigma(T)\}| \leq b(n-k)$ for any $\Sigma \subseteq \Sigma_0$ with~$|\Sigma|=k$.
In other words, from Equation~(\ref{bo}), $$N_k(T_0) \leq \binom{i_0}{k} b(n-k).$$
Consequently,  $$\pi_k(T_0) = N_k(T_0)/b(n) =  \binom{i_0}{k} b(n-k)/b(n) \leq  \binom{n-3}{k} b(n-k)/b(n),$$
where the last inequality holds because $i_0 \leq n-3$.\\

Finally, notice that we can write $$\binom{n-3}{k} b(n-k)/b(n) = \frac{1}{k!} \cdot \prod_{j=0}^{k-1} \frac{n-j-3}{2n-2j-5},$$
and each of the $k$ terms in the product is strictly less than~$\frac{1}{2}$. This completes the proof.\\
\end{proof}

\newpage

\section{Nearest Neighbour Interchange}

In this section we provide a new asymptotic expression for the size of the $k^{th}$ Nearest Neighbour Interchange (NNI) neighbourhood of an unrooted binary tree. \\

Let $T\in UB(n)$ and let $e=\{x,y\}$ be an interior edge of~$T$. Let $A_1$ and $A_3$ be subtrees of $T$ that are distance one from $e$ and distance three apart (see Fig.~\ref{NNI}). Then $A_1$ and $A_3$ are {\it swappable} across~$e$. Let vertex $z_1$ adjacent to $x$ be the root of~$A_1$, and $z_3$ adjacent to $y$ be the root of~$A_3$. An {\it NNI operation} on $T$ is performed by deleting the edges $\{x,z_1\}$ and $\{y,z_3\}$, and inserting edges $\{x,z_3\}$ and~$\{y,z_1\}$. We will also refer to this process as {\it swapping} the subtrees $A_1$ and $A_3$ {\it across}~$e$. The resulting tree is a {\it first NNI neighbour} of~$T$. To make it clear which edge of a tree $T$ two subtrees are swapped across in an NNI operation on~$T$, we will refer to such an operation as an {\it NNI operation on edge $e$ in~$T$}. \\

The two distinct first NNI neighbours resulting from an NNI operation on edge $e$ in $T$ can be seen in Fig.~\ref{NNI}. We have four subtrees $A_1$, $A_2$, $A_3$ and $A_4$ that are all distance one from $e$ in~$T$. To obtain $T'$ from $T$ we swap subtrees $A_2$ and~$A_3$, and to obtain $T''$ we swap subtrees $A_2$ and~$A_4$. Note that swapping subtrees $A_1$ and $A_4$ in $T$ produces a tree isomorphic to~$T'$. Although there are four different pairs of subtrees that could be swapped across~$e$, there are only two distinct first neighbours that can be obtained from NNI operations on~$e$. \\

\begin{figure}[htb]
\centering
\epsfig{file = 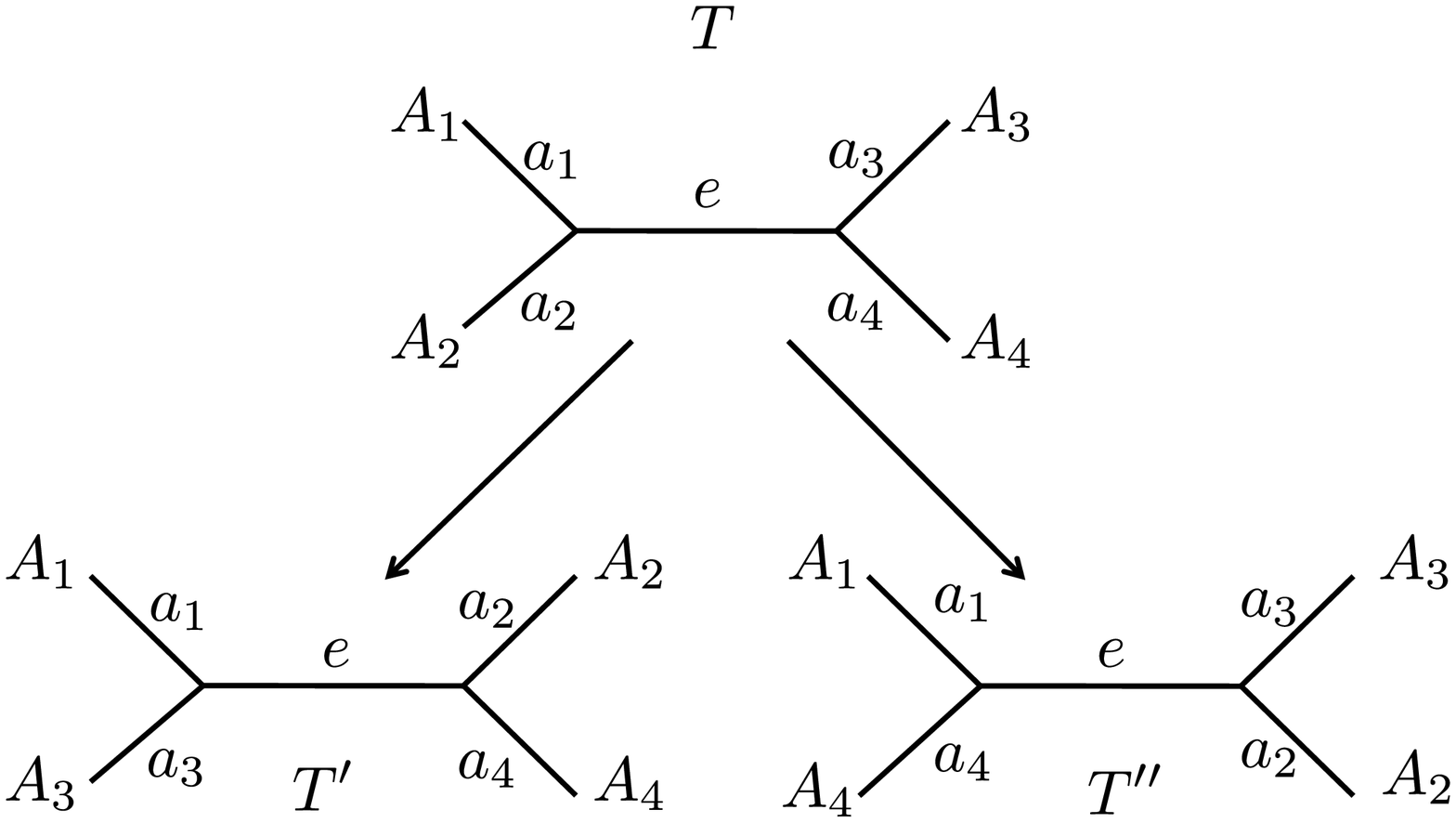, width = 0.5\linewidth}
\caption{The two first NNI neighbours of $T$ resulting from an NNI operation on the edge~$e$.}
\label{NNI}
\end{figure}

We see in Fig.~\ref{NNI} that in $T'$ and~$T''$, all four subtrees $A_1$, $A_2$, $A_3$ and $A_4$ are distance one from~$e$. Given a labelling of the edges of the original tree~$T$, we preserve this labelling by assigning the label $a_i$ to the edge incident to subtree~$A_i$ in $T$ and in the two first NNI neighbours of $T$ resulting from an NNI operation on edge~$e$. Note that $T$ can also be obtained from $T'$ by an NNI operation, which we call the {\it inverse} of the operation used to obtain $T'$ from~$T$.  \\

Consider a graph G in which each vertex represents a tree in $UB(n)$ and there is an edge between the vertices representing trees $T_1$ and $T_2$ if they are first NNI neighbours. The {\it NNI distance} between $T_1$ and~$T_2$, $\delta_{NNI}(T_1,T_2)$, is the distance between the two vertices representing trees $T_1$ and $T_2$ in~$G$. \\

Throughout this section we will consider the trees resulting from a series of NNI operations beginning with a tree~$T\in UB(n)$. Let $NNI(T; e_1,e_2,...,e_k)\subseteq \cup_{j=0}^kN^j_{NNI}(T)$ be the set of trees that can be obtained by performing an NNI operation on internal edge $e_1$ in $T$ to give~$T_1$, followed by an NNI operation on internal edge $e_2$ in $T_1$ to give~$T_2$, and so on until we have completed $k$ NNI operations. Note that if $T'\in NNI(T;e_1,...,e_k)$, $T'$ is not necessarily a $k^{th}$ NNI neighbour of~$T$. It may instead be a $j^{th}$ NNI neighbour of $T$ for some $j<k$ ($j\in\mathbb{N}$).\\

In this section we consider the sizes of the first, second, third and $k^{th}$ NNI neighbourhoods of an unrooted binary phylogenetic tree. Let $T\in UB(n,c)$ (recall that $UB(n,c)$ is the set of unrooted, binary phylogenetics trees with $n$ leaves and $c$ cherries). Then
\begin{enumerate}
\item $|N_{NNI}(T)|=2(n-3)$, and
\item $|N^2_{NNI}(T)|=2n^2-10n+4c$, and
\item $|N^3_{NNI}(T)|=\frac{4}{3}n^3-8n^2-\frac{70}{3}n+8cn+12p_3(T)+164,$
\end{enumerate}
where $p_3(T)$ is the number of internal paths of length three in~$T$. These results for the first and second NNI neighbourhoods were shown by Robinson~\cite{Robinson}. It is interesting to compare these with the corresponding results for the RF distance. In both cases, the first neighbourhood is dependent only on the number of leaves, while the second neighbourhood is determined by the number of leaves and cherries. In fact, the size of the first NNI neighbourhood is the same as the size of the first RF neighbourhood. Robinson~\cite{Robinson} also found an upper bound on the size of the third NNI neighbourhood of a binary phylogenetic tree. We use this result to derive the exact formula for the size of the third NNI neighbourhood given above (see the proof in Appendix A, along with a brief discussion of how to determine~$p_3(T)$.\\

As mentioned previously, tree rearrangement operations are also used to compare trees produced by different tree reconstruction methods, or trees obtained from different data sets. This can be achieved by determining the NNI distance (the smallest number of operations) between the two trees. DasGupta {\it et al.}~\cite{DasGupta} showed that the problem of computing the NNI distance between two trees in $UB(n)$ is NP-complete. Culik and Wood~\cite{Culik} found an upper bound of $4n\log(n)$ on the NNI distance between two trees in $UB(n)$, which was later improved to $n\log(n)$ by Li {\it et al.}~\cite{Li}. \\

It is also useful to understand the structure of $UB(n)$, and the first and second NNI neighbourhoods of a tree (e.g. how the first NNI neighbours of a tree relate to each other). A {\it walk} in a graph G is a sequence of vertices and edges, in which the vertices are not necessarily distinct. Consider a graph G in which each vertex represents a tree in $UB(n)$ and there is an edge between the vertices representing trees $T_1$ and $T_2$ if they are first NNI neighbours. Bryant~\cite{Challenge} noted that the length of the shortest walk that visits every vertex of $G$ was unknown. Gordon {\it et al.}~\cite{Hamilton} provided a constructive proof that this walk is a Hamiltonian path (a path that visits every vertex of $G$ exactly once). Therefore by a series of NNI operations beginning from a tree $T\in UB(n)$, it is possible to visit each tree in $UB(n)$ exactly once. We refer to this series of NNI operations as an {\it NNI walk}. In Section~5, we investigate the structure of $UB(n)$ by determining the number of pairs of trees that share a first NNI neighbour (the number of pairs of trees that are within NNI distance two of each other). 

\subsection{Asymptotic Result for the $k^{th}$ Neighbourhood}

Our main result for this section is the asymptotic expression for the size of the $k^{th}$ NNI neighbourhood of a binary tree given in Theorem~\ref{main}. \\

\begin{theorem}
Let $T\in UB(n)$ ($n\geq 4$). Then for each fixed $k\in\mathbb{Z}^+$,
\begin{equation}|N^k_{NNI}(T)|=\frac{2^kn^k}{k!}\left(1+D_{T,k}n^{-1}+O(n^{-2})\right),\label{NNIEqn}\end{equation}
where 
$$-\frac{3k(k+1)}{2} \leq D_{T,k} \leq 3k(k-2).$$
\label{main}
\end{theorem}

We will prove Theorem~\ref{main} at the end of Section~4.1. As in the proof of Theorem~\ref{mainRF}, we consider the number of $k^{th}$ NNI neighbours resulting from NNI operations on a given set of $k$ internal edges. From Lemma~\ref{counting}, we know the number sets of $k$ internal edges of~$T$. Combining these gives us the total number of $k^{th}$ NNI neighbours. The four different cases that are relevant are:
\begin{enumerate}
\item The $k$ edges are distinct and pairwise non-adjacent.
\item The $k$ edges are distinct and exactly two are adjacent.
\item The $k$ edges are distinct and more than two are adjacent. 
\item The $k$ edges are not all distinct.\\
\end{enumerate}
These are the same cases as for RF, with the additional possibility that the $k$ edges are not all distinct (Case $4$). In Equation~\ref{NNIEqn} of Theorem~\ref{main}, the term of order $n^k$ is completely determined by Case $1$, whereas the term of order $n^{k-1}$ is determined by Cases $1$ and~$2$. We show that all other possibilities for the $k$ edges (covered by Cases $3$ and $4$) only contribute to terms of order $n^{k-2}$ or lower. \\

\textbf{Neighbours Resulting from NNI Operations on $k$ Given Edges}\\

We cannot determine the size of the $k^{th}$ NNI neighbourhood by simply counting the number of possible sets of $k$ NNI operations, as there are situations where changing the order of the operations, or the edges they are on, does not change the resulting tree. We need to determine exactly when this occurs in order to accurately count the size of the $k^{th}$ neighbourhood. Therefore, some preliminary work is required before we investigate the four cases outlined above.\\

First, we consider the impact that a single NNI operation on a tree $T$, has on the non-trivial splits of~$T$. \\

\begin{lemma}
Let $T\in UB(n)$ ($n\geq 4$) and $T'=NNI(T;e)$ where $e$ is an internal edge of~$T$. Then
$$|\Sigma(T)-\Sigma(T')|=|\Sigma(T')-\Sigma(T)|=1.$$
Furthermore, $\Sigma(T)-\Sigma(T')=\{S(T,e)\}$, and for all internal edges $e'\neq e$ in~$T$, we have \text{$S(T',e')=S(T,e')$.} \\
\label{sigma}
\end{lemma}
\begin{proof}
Note that $|\Sigma(T)|=|\Sigma(T')|$ as $T,T'\in UB(n)$. Let the subtrees that are distance one from $e$ in $T$ be called $A$, $B$, $C$ and $D$, such that $d_T(A,B)=2$. Then we have \text{$S(T,e)=\{\mathcal{L}(A)\cup \mathcal{L}(B),\mathcal{L}(C)\cup \mathcal{L}(D)\}$.} Either $A$ or $B$ is one of the subtrees that are swapped by the NNI operation, so $d_{T'}(A,B)=3$, and $\mathcal{L}(A)$ and $\mathcal{L}(B)$ are in different parts of $S(T',e)$. Hence $S(T,e)\neq S(T',e)$.\\

Suppose there exists an internal edge $e'$ of~$T$, such that $e'\neq e$. Let $S(T,e')=\{L_1,L_2\}$. In~$T$, either $e'$ is adjacent to~$e$, or $e'$ is in one of the subtrees $A$, $B$, $C$ or~$D$. Therefore, either $L_1$ or $L_2$ is a subset of the leaves in one of the subtrees $A$, $B$, $C$ or~$D$. Since $A$, $B$, $C$ and $D$ are the four subtrees of $T'$ that are distance one from~$e$, $S(T',e')=S(T,e')$. \\

Therefore $\Sigma(T)-\Sigma(T')=\{S(T,e)\}$ and $\Sigma(T')-\Sigma(T)=\{S(T',e)\}$. Hence
$$|\Sigma(T)-\Sigma(T')|=|\Sigma(T')-\Sigma(T)|=1.$$
\end{proof} 

We now compare the non-trivial splits of trees that are $k^{th}$ NNI neighbours. \\

\begin{lemma}
Let $T\in UB(n)$ ($n\geq 4$), and let $e_1$, ..., $e_k$ ($k\geq 1$) be internal edges of $T$ such that there exists $e_m$ ($1\leq m\leq k$) for which $e_m\not\in\{e_1,...,e_{m-1},e_{m+1},...,e_k\}$. Let \text{$T'\in NNI(T;e_1,...,e_k)$}. Then
\begin{enumerate}
\item[(i)] if $e'$ is an internal edge of~$T$ and $e'\not\in\{e_1,...,e_k\}$, then $S(T',e')=S(T,e')$, and
\item[(ii)] $S(T, e_m)$ is not a split of~$T'$.\\
\end{enumerate}
\label{extra1}
\end{lemma}
\begin{proof}\mbox{}
\begin{enumerate}
\item[(i)] To see that $S(T',e')=S(T,e')$, consider trees $T_1'$, ..., $T_{k-1}'$ where $T_1'\in NNI(T;e_1)$, \text{$T_2'\in NNI(T_1';e_2)$},~..., $T_{k-1}'\in NNI(T_{k-2}';e_{k-1})$, and $T'\in NNI(T_{k-1}';e_k)$. Then by Lemma~\ref{sigma}, 
$$S(T',e')=S(T_{k-1}',e')=\cdots=S(T_1',e')=S(T,e').$$

\item[(ii)] Let $T_{m-1}\in NNI(T;e_1,...,e_{m-1})$ such that $T'\in NNI(T_{m-1};e_m,...,e_k)$. By $(i)$, since \text{$e_m\not\in\{e_1,...,e_{m-1}\}$}, $S(T_{m-1},e_m)=S(T,e_m)$. Let $A$ and $B$ be two subtrees in $T_{m-1}$ that are distance one from~$e_m$, where $d(A,B)=2$. Then
$$S(T_{m-1},e_m)=S(T,e_m)=\{\mathcal{L}(A)\cup \mathcal{L}(B),\mathcal{L}(T)-(\mathcal{L}(A)\cup \mathcal{L}(B))\}.$$
The NNI operation on edge $e_m$ in $T_{m-1}$ swaps either $A$ or $B$ with one of the other two subtrees that are distance one from~$e$. Let $T_m\in NNI(T_{m-1};e_{m})$ such that $T'\in NNI(T_m;e_{m+1},...,e_k)$. Then $\mathcal{L}(A)$ and $\mathcal{L}(B)$ are in different parts of $S(T_m,e_m)$. Since $e_m\not\in\{e_{m+1},...,e_k\}$, we have $S(T',e_m)=S(T_m,e_m)$ by $(i)$. Hence all leaves in $A$ are in a different component of $T'\setminus e_m$ to the leaves of~$B$. Therefore, $S(T,e_m)$ is not a split of~$T'$. \\
\end{enumerate}
\end{proof}

\begin{corollary}
Let $T\in UB(n)$ ($n\geq 4$) and let $e_1$, ..., $e_k$ ($k\geq 1$) be internal edges of $T$ such that there exists $e_m$ ($1\leq m\leq k$) for which $e_m\not\in\{e_1,...,e_{m-1},e_{m+1},...,e_k\}$. Let $P=NNI(T;e_1,...,e_j)$ and $Q=NNI(T;e_{j+1},...,e_k)$ where $1\leq j\leq k$. Then $$P\cap Q=\emptyset.$$
\label{splits}
\end{corollary}
\begin{proof}
Without loss of generality, suppose that $1\leq m\leq j$. By Lemma~\ref{extra1}, $S(T, e_m)$ is not a split of any of the trees in~$P$.\\

Additionally by Lemma~\ref{extra1}, for all $T'\in Q$, $S(T',e_m)=S(T,e_m)$, since $e_m\not\in\{e_{j+1},...,e_k\}$. Therefore, since two trees are equal if and only if they have the same set of splits, we have $P\cap Q=\emptyset$.\\
\end{proof}

We now consider the number of $k^{th}$ NNI neighbours of a tree resulting from a series of $k$ NNI operations on a given set of distinct edges.\\

\begin{lemma}
Let $T\in UB(n)$ ($n\geq 4$), and let $e_1$, $e_2$,..., $e_k$ ($1\leq k\leq n-3$) be distinct internal edges of~$T$. Then 
$NNI(T;e_1,...,e_k)$ is a subset of $N^k_{NNI}(T)$ of size~$2^k$.\\
\label{2k}
\end{lemma}
\begin{proof}
For each edge $e$ of $T$ there are two distinct first NNI neighbours resulting from NNI operations on~$e$. Since we perform NNI operations on $k$ different edges in~$T$, there are $2^k$ $k^{th}$ NNI neighbours, provided that none of the resulting trees are equivalent, or in the $j^{th}$ NNI neighbourhood of $T$ for some~$j<k$. \\

The latter follows from Corollary~\ref{splits}, since $e_1,...,e_k$ are distinct. This means that
$$NNI(T;e_1,...,e_k)\subseteq N^k_{NNI}(T).$$

To show that none of the resulting $2^k$ trees are equivalent, we consider the splits of these trees. Let $T_k$ and $T_k'$ be two trees in $NNI(T;e_1,...,e_k)$, where at least one operation produced a different first neighbour in each case. In other words, there exist trees $T_{m-1}, T_m, T_m'\in UB(n)$ such that \text{$T_{m-1}\in NNI(T;e_1,...,e_{m-1})$}, \text{$\{T_m, T_{m}'\}\subseteq NNI(T_{m-1};e_m)$}, \text{$T_k\in NNI(T_m;e_{m+1},...,e_k)$}, \text{$T_k'\in NNI(T_m';e_{m+1},...,e_k)$}, and $T_m\neq T_{m+1}$. Note that since $T_m$ and $T_m'$ are the two distinct first NNI neighbours of $T_{m-1}$ obtained by an NNI operation on~$e_m$, $T_m'\in NNI(T_m;e_m)$.   \\

Now we consider the splits of $T$, $T_m$, $T_m'$, $T_k$, and~$T_k'$. By Lemma~\ref{extra1}, $S(T_m,e_m)\neq S(T_m',e_m)$, as $T_m$ and $T_m'$ are first NNI neighbours (by an operation on edge $e_m$). Since $e_m\not\in \{e_{m+1},...,e_k\}$, $S(T_k,e_m)=S(T_m,e_m)\neq S(T_m',e_m)=S(T_k',e_m)$ by Lemma~\ref{extra1}. Hence $T_k\neq T_k'$. Therefore, we obtain $2^k$ distinct $k^{th}$ NNI neighbours from $k$ NNI operations on distinct edges $e_1$,...,$e_k$ in order.\\
\end{proof}

An important factor to consider, is the effect that the distance between two edges in consecutive NNI operations on a tree $T$ has on the the resulting second NNI neighbours. The following result is due to Robinson~\cite{Robinson} and was originally proved by exhaustion, leaving the details to the reader. We provide an alternate proof in Appendix B. \\

\begin{lemma}
Let $T\in UB(n)$ ($n\geq 5$), and let $e_1$ and $e_2$ be distinct internal edges of~$T$. Let $P=NNI(T; e_1, e_2)$ and $Q=NNI(T; e_2, e_1)$. If $e_1$ and $e_2$ are adjacent then $P\cap Q=\emptyset$; otherwise,~$P=Q$. \\
\label{two}
\end{lemma}

This leads naturally to the following corollary.\\

\begin{corollary}
Let $T\in UB(n)$ ($n\geq 4$), and let $e_1$, ..., $e_k$ ($k\geq 1$) be internal edges of~$T$. Let 
$$P=NNI(T;e_1,...,e_m,e_{m+1},...,e_k),$$ 
$$Q=NNI(T;e_1,...,e_{m+1},e_m,...,e_k).$$
If $e_m$ and $e_{m+1}$ are distinct and non-adjacent, then~$P=Q$. \\
\label{swap}
\end{corollary}
\begin{proof}
Let $T_{m-1}$ be a tree in $NNI(T;e_1,...,e_{m-1})$. By Lemma~\ref{two}, $$NNI(T_{m-1};e_m,e_{m+1})=NNI(T_{m-1};e_{m+1},e_m).$$ This is true for any choice of $T_{m-1}$, so $$NNI(T;e_1,...,e_{m-1},e_m,e_{m+1})=NNI(T;e_1,...,e_{m-1},e_{m+1},e_m).$$ Therefore~$P=Q$.\\
\end{proof}

Lemma~\ref{two} and Corollary~\ref{swap} tell us how the distance between the two edges in successive NNI operations affects the resulting second NNI neighbours. Corollary~\ref{splits} justifies that different choices of edges for the two NNI operations do not produce any duplicate second NNI neighbours. Lemma~\ref{2k} tells us the number of second NNI neighbours resulting from NNI operations on a given set of internal edges of a tree in order. We now have enough information to present some results on the number of neighbours resulting from NNI operations on a given (unordered) set of internal edges.\\

\begin{lemma}
Let $T\in UB(n)$ ($n\geq 4$). 
\begin{enumerate}
\item[(i)] For any given set of $k$ distinct, pairwise non-adjacent internal edges ($1\leq k\leq n-3$), there are $2^k$ $k^{th}$ neighbours of $T$ resulting from NNI operations on this sequence of edges in any order.
\item[(ii)] For any given set of $k$ distinct internal edges ($2\leq k\leq n-2$) where exactly one pair is adjacent, there are $2^{k+1}$ $k^{th}$ neighbours of $T$ resulting from NNI operations on this sequence edges in any order.
\item[(iii)] For a given $T$ and a given sequence of $k$ (not necessarily distinct) edges of $T$ ($k\geq 1$), the number of $k^{th}$ NNI neighbours resulting from NNI operations on this sequence edges in any order is constant with respect to~$n$.\\
\end{enumerate}
\label{const}
\end{lemma}
\begin{proof}
\mbox{}
\begin{enumerate}
\item[(i)] Suppose we perform NNI operations on $k$ distinct, pairwise non-adjacent internal edges \text{$e_1$, ..., $e_k$} of~$T$. Lemma~\ref{2k} tells us that if the NNI operations are performed in a given order we obtain $2^k$ $k^{th}$ neighbours. Since the edges are pairwise non-adjacent, by Corollary~\ref{swap}, changing the order of the operations does not change the set of trees produced. Hence, there are $2^k$ neighbours of $T$ resulting from NNI operations on this set of edges in any order.\\

\item[(ii)] The only difference between this and (i) is the pair of adjacent edges $e_i$ and $e_j$ (where $1\leq i<j\leq k$).  By repeated applications of Corollary~\ref{swap}, 
$$NNI(T;e_1,...,e_i,...,e_j,...,e_k)=NNI(T;e_1,...,e_{i-1},e_{i+1},...,e_{j-1},e_{j+1},...,e_k,e_i,e_j).$$ 
As in (i), by Lemma~\ref{2k} and Corollary~\ref{swap}, performing NNI operations on the edges \text{$e_1$, ..., $e_{i-1}$, $e_{i+1}$, ..., $e_{j-1}$ , $e_{j+1}$, ..., $e_k$} in any given order produces the set of trees $$NNI(T;e_1,...,e_{i-1},e_{i+1},...,e_{j-1},e_{j+1},...,e_k),$$ where 
$$|NNI(T;e_1,...,e_{i-1},e_{i+1},...,e_{j-1},e_{j+1},...,e_k)|=2^{k-2}.$$
Let $T_{k-2}\in NNI(T;e_1,...,e_{i-1},e_{i+1},...,e_{j-1},e_{j+1},...,e_k)$. By Lemma~\ref{two},
$$NNI(T_{k-2};e_i,e_j)\cap NNI(T_{k-2};e_j,e_i)=\emptyset.$$ 
Therefore, since
$$|NNI(T_{k-2};e_i,e_j)\cup NNI(T_{k-2};e_j,e_i)|=4+4=8,$$ we have $8(2^{k-2})=2^{k+1}$, $k^{th}$ NNI neighbours of~$T$.\\

\item[(iii)] Let $F$ be the subgraph of $T$ consisting of the edges $e_1,...,e_k$. Then $F$ has $m$ components, $C_1,...,C_m$ ($1\leq m \leq k$). Edges in different components of $F$ are not adjacent, so by Corollary~\ref{swap}, the order in which we perform NNI operations on them does not change the resulting neighbours. However, by Lemma~\ref{two}, the order of NNI operations on the edges that form a component of $F$, does change the resulting  neighbours. Therefore the number of neighbours resulting from NNI operations on the $k$ edges is
$$\prod_{\ell=1}^m f(C_{\ell}),$$
where $f(C_{\ell})$ is the number of distinct $k^{th}$ NNI neighbours resulting from NNI operations in $T$ on the edges from $e_1,...,e_k$ that are in component $C_{\ell}$ of $F$ (more than one NNI operation may be on the same edge). We consider each component separately. \\

Let $C_p$, $1\leq p\leq m$ be a component of $F$ with $q$ edges and consider calculating~$f(C_p)$. Let $f_1,...,f_j$ ($j\leq k$) be the subsequence of the edges $e_1,...,e_k$ that are in~$C_p$. Note that the edges $f_1,...,f_j$ are not necessarily distinct. Add pendant edges incident to vertices in $V(C_p)$, so that all of the vertices in $V(C_p)$ have degree three. The resulting tree $C_p'$ is an unrooted binary tree with $q+3$ leaves. The internal edges of $C_p'$ are the distinct edges of the sequence $f_1,...,f_j$. Then $f(C_p)$ is equivalent to the number of distinct $k^{th}$ NNI neighbours of $C_p'$ resulting from NNI operations on the edges $f_1$, ..., $f_j$ of~$C_p'$. The number of $k^{th}$ neighbours $f(C_p')$ from these operations depends only on the shape and size of $C_p'$, the number of times we perform an NNI operation on each internal edge of $C_p$, and the order in which the operations are performed. All of these factors are determined by the choice of the edges $e_1,...,e_k$ of~$T$. Therefore, given a tree $T$ and internal edges $e_1,...,e_k$ of $T$, the number of $k^{th}$ NNI neighbours of $T$ resulting from NNI operations on the edges $e_1,...,e_k$ in any order is independent of~$n$.   
\end{enumerate}
\end{proof}

Before we consider the number of sets of $k$ edges we need two more results, which tell us that the trees resulting from $k$ NNI operations on less than $k$ pairwise non-adjacent internal edges of a tree, are not $k^{th}$ neighbours (the case where the $k$ edges are not all distinct). This is important, as if some of these trees are $k^{th}$ neighbours, then this case contributes to the $O(n^{-1})$ term in Theorem~4.1. First, in Lemma~\ref{same}, we consider two consecutive NNI operations on the same edge. Then, in Corollary~\ref{simplify} we consider two non-consecutive NNI operation on the same edge.\\

\begin{lemma}
Let $T\in UB(n)$ ($n\geq 4$), and let $e_1$,..., $e_k$ ($k\geq 2$) be internal edges of~$T$. Suppose there exists an m ($1\leq m\leq k-1$) for which $e_m=e_{m+1}$. Let
$$P=NNI(T;e_1,...,e_m,e_{m+1},e_{m+2},...,e_k).$$
\begin{enumerate}
\item[(i)] If the NNI operation on edge $e_{m+1}$ is the inverse of the operation on edge $e_m$, then
$$P=NNI(T;e_1,...,e_{m-1},e_{m+2},...,e_k).$$
\item[(ii)] If the NNI operation on edge $e_{m+1}$ is not the inverse of the operation on edge $e_m$, then
$$P=NNI(T;e_1,...,e_m,e_{m+2},...,e_k).$$
\end{enumerate}
Furthermore, $P\cap N_{NNI}^k(T)=\emptyset$.\\
\label{same}
\end{lemma}
\begin{proof}\mbox{}
Let $T_{m-1}\in NNI(T;e_1,...,e_{m-1})$ and let $T_m,T_m'\in NNI(T_{m-1};e_m)$, $T_m\neq T_m'$. Now suppose we perform an operation on edge $e_{m+1}=e_m$ in $T_m$ to obtain a tree~$T_{m+1}$. Let \\$T'\in NNI(T_{m+1};e_{m+2},...,e_k)$. \\

\begin{enumerate}
\item[(i)] First, suppose the operation on edge $e_{m+1}$ is the inverse of the operation on edge~$e_m$. Then \text{$T_{m+1}=T_{m-1}$.} Hence $P=NNI(T;e_1,...,e_{m-1},e_{m+2},...,e_k).$\\

\item[(ii)] Now suppose that the operation on edge $e_{m+1}$ is not the inverse of the operation on edge~$e_m$. Then \text{$T_{m+1}=T_m'$.} Hence $P=NNI(T;e_1,...,e_m,e_{m+2},...,e_k).$\\
\end{enumerate}
It follows from $(i)$ and $(ii)$ that $P\cap N_{NNI}^k(T)=\emptyset$.\\
\end{proof}

\begin{corollary}
Let $T\in UB(n)$ ($n\geq 4$), and let $e_1$, ..., $e_k$ ($k\geq 2$) be internal edges of~$T$. Suppose that $e_m=e_j$ for some $m, j$, where $1\leq m<j\leq k$. Let 
$$P=NNI(T;e_1,...,e_m,e_{m+1},...,e_{j-1},e_j,...,e_k),$$
$$Q=NNI(T;e_1,...,e_{m-1},e_{m+1},...,e_{j-1},e_j,...,e_k)$$
$$R=NNI(T;e_1,...,e_{m-1},e_{m+1},...,e_{j-1},e_{j+1},...,e_k).$$
Suppose that the edges $e_{m+1},...,e_{j-1}$ are non-adjacent to~$e_m$. If the operation on edge $e_j$ is the inverse of the operation on edge $e_m$, then $P=R$; otherwise,~$P=Q$.\\
\label{simplify}
\end{corollary}
\begin{proof}
By Corollary~\ref{swap},
\begin{align*}
P
&=NNI(T;e_1,...,e_{m-1},e_{m+1},e_m,...,e_{j-1},e_j,...,e_k)\\
&=NNI(T;e_1,...,e_{m-1},e_{m+1},e_{m+2},e_m,...,e_{j-1},e_j,...,e_k)\\
\vdots\\
&=NNI(T;e_1,...,e_{m-1},e_{m+1},...,e_{j-1},e_m,e_j,...,e_k).
\end{align*}
Therefore, by Lemma~\ref{same}, if the operation on edge $e_j$ is the inverse of the operation on edge $e_m$, $P=R$; otherwise,~$P=Q$.\\
\end{proof}

Now we have all of the information required to prove Theorem~\ref{main}.

\subsubsection*{Proof of Theorem \ref{main}}

We break down the calculation of the size of the $k^{th}$ NNI neighbourhood of $T$ into two steps. First we consider the number of $k^{th}$ NNI neighbours resulting from $k$ NNI operations on a given sequence of $k$ edges of~$T$. We then consider the number of ways these $k$ edges can be chosen in~$T$. By Lemma~\ref{const}, the number of $k^{th}$ NNI neighbours of a given tree $T$ resulting from operations on a given sequence of $k$ edges is not dependent on~$n$. Hence, only the number of ways of choosing these $k$ edges is dependent on $n$. We consider two cases. \\

First, assume that the $k$ edges are all distinct, and consider the number of ways they can be chosen in~$T$. By Lemma~\ref{counting} the case where the $k$ edges are pairwise non-adjacent (Case 1 from the beginning of this subsection) gives a term of order $n^k$ and a term of order~$n^{k-1}$. The case where exactly two of the $k$ edges are adjacent (Case 2) produces a term of order $n^{k-1}$, but not a term of order~$n^k$. If more than two of the $k$ edges are adjacent then the highest order term is~$O(n^{k-2})$. \\ 

Now suppose that the $k$ edges are not all distinct. By Lemma~\ref{counting}, if $k-1$ of the $k$ edges are distinct and pairwise non-adjacent, the highest order term is~$O(n^{k-1})$.  However, by Corollary~\ref{simplify}, the trees produced by this are not $k^{th}$ NNI neighbours of~$T$. By Lemma~\ref{counting}, if more than two of the $k$ edges are the same or if more than two are adjacent, the highest order term is~$O(n^{k-2})$. \\

In the case where the edges are pairwise non-adjacent, by Lemma~\ref{const}, there are $2^k$ $k^{th}$ NNI neighbours of $T$ resulting from NNI operations on a given set of $k$ edges. In the case where exactly two edges are adjacent there are $2^{k+1}$ resulting $k^{th}$ NNI neighbours. Hence by Lemma~\ref{counting},
\begin{align*}
|N^k_{NNI}(T)|&\geq \left(\frac{1}{k!}n^k-\frac{k(5k+1)}{2k!}n^{k-1}\right)2^k+\frac{1}{2(k-2)!}n^{k-1}2^{k+1}+O(n^{k-2})\\
&=\frac{2^k}{k!}n^k-\frac{3k(k+1)}{2k!}2^kn^{k-1}+O(n^{k-2});\\
\end{align*}
\begin{align*}
|N^k_{NNI}(T)|&\leq \left(\frac{1}{k!}n^k-\frac{k(k+2)}{k!}n^{k-1}\right)2^k+\frac{2}{(k-2)!}n^{k-1}2^{k+1}+O(n^{k-2})\\
&=\frac{2^k}{k!}n^k+\frac{3k(k-2)}{k!}2^kn^{k-1}+O(n^{k-2}).\\
\end{align*}
\tallqed

We can see that this result is very similar to the size of the $k^{th}$ RF neighbourhood, as $D_{T,k}$ and $C_{T,k}$ are both quadratic in~$k$. \\

\section{Pairs of Trees with Shared Neighbours}

In this section we calculate the number of pairs of binary phylogenetic trees with $n$ leaves that share a first NNI or RF neighbour. Equivalently,  the number of pairs of trees that are within at most distance two of each other. \\

Our calculation involves summing the size of the first and second neighbourhoods of a tree, over all binary phylogenetic trees, and discounting any duplicate trees. However, the size of the second neighbourhood for both NNI and RF is dependent on the number of cherries, by Bryant and Steel~\cite{Dist} and Robinson~\cite{Robinson}. Therefore it is necessary to know the number of binary phylogenetic trees with $n$ leaves and $c$ cherries,~$|UB(n,c)|$. Hendy and Penny~\cite{HendyPenny} found an expression for $|UB(n,c)|$, which they proved using induction on the number of leaves. Here we present a constructive proof of their result.\\

\begin{proposition}For all $n\geq 4$, 
$$|UB(n,c)|= \frac{n!(n-4)!}{c!(c-2)!(n-2c)!2^{2c-2}},$$
for $2\leq c\leq \frac{n}{2}$, and $|UB(n,c)|=0$ otherwise. \\
\label{hendypenny}
\end{proposition}
\begin{proof}
The tree with the smallest number of cherries is a caterpillar, which has two cherries. Since there are two leaves in a cherry, the maximum number of cherries a tree can have is~$\frac{n}{2}$. Hence for $c<2$ or $c>\frac{n}{2}$ we have $|UB(n,c)|=0$. \\

Let $2\leq c\leq \frac{n}{2}$. Each $T\in UB(n,c)$ has $2c$ leaves that are in cherries. The number of ways of choosing the $2c$ leaves of $T$ to form the $c$ cherries is~${{n}\choose{2c}}$. From those $2c$ leaves we choose two for each cherry. We divide by $c!$ since the ordering of the cherries is not important. (Note that this is the same as the number of perfect matchings on a complete graph with $2c$ vertices.) Therefore, the number of ways of choosing $c$ cherries from $n$ leaves is
$$M={{n}\choose{2c}}\frac{(2c)!}{c!2^c}=\frac{n!}{c!(n-2c)!2^c}.$$

Now consider each cherry as a single leaf with the labels of both leaves. There are $c$ of these double-labelled leaves and $n-2c$ other leaves. We determine the number of trees that can be formed with these leaves. We have the restriction that no pair of the $n-2c$ single-labelled leaves can be in a cherry. Therefore, we will first consider the number of trees we can form with only the $c$ double-labelled leaves. This number, $P$, is given in Lemma~\ref{intedges},
$$P=|UB(c)|=\frac{(2c-4)!}{(c-2)!2^{c-2}}.$$

Now  let $T$ be one of these trees with $c$ double-labelled leaves. We insert the remaining $n-2c$ single-labelled leaves. Each single-labelled leaf can only be joined to edges in $E(T)$, so as not to create another cherry. There are $2c-3$ edges in $E(T)$ to which the single-labelled leaves could be joined. Since there are no other restrictions on where these single-labelled leaves must be inserted, we simply need to count the number of distinct trees resulting from joining the $n-2c$ single-labelled edges to edges in~$E(T)$. This is the product of the number of ways to place $n-2c$ items into $2c-3$ bins, and the number of ways to order the $n-2c$ items. The number of distinct trees is given by
\begin{align*}
Q&=(n-2c)!{{(n-2c)+(2c-3)-1}\choose{(2c-3)-1}}\\
&=(n-2c)!{{n-4}\choose{2c-4}}=\frac{(n-4)!}{(2c-4)!}.
\end{align*}

Combining $M$, $P$, and $Q$, we have
\begin{align*}
|UB(n,c)|=MPQ&=\frac{n!}{c!(n-2c)!2^c}\cdot\frac{(2c-4)!}{(c-2)!2^{c-2}}\cdot\frac{(n-4)!}{(2c-4)!}\\
&=\frac{n!(n-4)!}{c!(c-2)!(n-2c)!2^{2c-2}}.
\end{align*}
\end{proof}

We can now use this result to find the number of pairs of binary phylogenetic trees in $UB(n)$ that are within at most distance two of each other under NNI and RF. For $\theta\in\{NNI, RF\}$, define

$$N^{\leq k}_{\theta}(n)=\{(T,T'):T,T'\in UB(n), d_{\theta}(T,T')\leq k\}.$$

\newpage
\begin{corollary}\mbox{}
Let $n\geq 3$, Then
\begin{enumerate}
\item[(i)] $|N^{\leq 2}_{NNI}(n)|=\sum_{c=2}^{\lfloor\frac{n}{2}\rfloor}|UB(n,c)|(n^2-4n+2c-3).$
\item[(ii)] $|N^{\leq 2}_{RF}(n)|=\sum_{c=2}^{\lfloor\frac{n}{2}\rfloor}|UB(n,c)|(n^2-3n+3c-9).$\\
\end{enumerate}
\end{corollary}
\begin{proof}\mbox{}
\begin{enumerate}
\item[(i)] For $T\in UB(n,c)$, the number of first and second NNI neighbours is
\begin{align*}
N_{NNI}(T)+N^2_{NNI}(T)&=2(n-3)+2n^2-10n+4c\\
&=2n^2-8n+4c-6.
\end{align*}

To find the number of pairs of trees in $UB(n)$ that are within NNI distance two, we simply sum the number of first and second neighbours over all trees in $UB(n)$, and then halve the result as each pair will be counted twice. So, 

\begin{align*}
|N^{\leq 2}_{NNI}(n)|&=\frac{1}{2}\sum_{c=2}^{\lfloor\frac{n}{2}\rfloor}|UB(n,c)|(2n^2-8n+4c-6)\\
&=\sum_{c=2}^{\lfloor\frac{n}{2}\rfloor}|UB(n,c)|(n^2-4n+2c-3).
\end{align*}

Proposition~\ref{hendypenny} gives us~$|UB(n,c)|$.\\

\item[(ii)] For each unrooted binary tree $T$, the number of first and second RF neighbours is
\begin{align*}
N_{RF}(T)+N^2_{RF}(T)&=2(n-3)+2n^2-8n+6c-12\\
&=2n^2-6n+6c-18.
\end{align*}
Therefore
\begin{align*}
|N^{\leq 2}_{RF}(n)|&=\frac{1}{2}\sum_{c=2}^{\lfloor\frac{n}{2}\rfloor}|UB(n,c)|(2n^2-6n+6c-18)\\
&=\sum_{c=2}^{\lfloor\frac{n}{2}\rfloor}|UB(n,c)|(n^2-3n+3c-9).
\end{align*}
\end{enumerate}
\end{proof}

\section{Subtree Prune and Regraft}

In this section, we show that unlike RF and NNI, the size of the second Subtree Prune and Regraft (SPR) neighbourhood of a tree $T\in UB(n)$ is not uniquely determined by the number of leaves and cherries of~$T$. \\

An {\it SPR operation} on a tree $T\in UB(n)$ is defined by the following process:
\begin{enumerate}
\item Choose an edge $e=\{u,v\}\in E(T)$ and delete it, leaving two components $T_u$ (containing the vertex $u$) and $T_v$ (containing the vertex $v$).
\item Choose an edge $f\in E(T_v)$ and subdivide $f$ with a new vertex $w$ to obtain two edges $f_1$ and~$f_2$. The vertex $w$ has degree two.
\item Insert the edge $g=\{w,u\}$ and suppress the vertex $v$ to obtain a binary tree~$T'\in UB(n)$. \\
\end{enumerate}

Essentially, we {\it prune} the subtree $T_u$ and {\it regraft} it onto edge~$f$. We refer to $e$ as the {\it cut edge} and $f$ as the {\it join edge} of the SPR operation (see Fig.~\ref{SPRDiagram}). The tree $T'$ is a {\it first SPR neighbour} of~$T$. We will use the notation $SPR(T, (e,f))$ to refer to the tree obtained by an SPR operation on tree $T$ with cut edge $e$ and join edge~$f$. Note that if $d_T(e,f)=1$, then $T'$ is a first NNI neighbour of~$T$~\cite{Book}. \\

\begin{figure}[htb]
\centering
\epsfig{file = 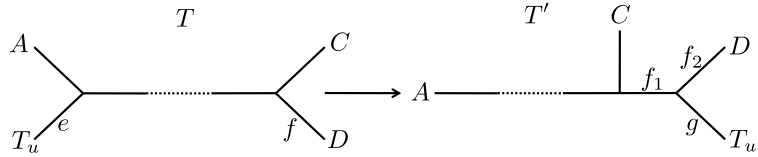, width = 0.7\linewidth}
\caption{An example of an SPR operation with cut edge $e$ and join edge~$f$.}
\label{SPRDiagram}
\end{figure}

Consider a graph $G$ in which each vertex represents a tree in $UB(n)$ and there is an edge between the vertices representing trees $T_1$ and $T_2$ if they are first SPR neighbours. The {\it SPR distance} between $T_1$ and $T_2$, $\delta_{SPR}(T_1,T_2)$, is the distance between the two vertices representing $T_1$ and $T_2$ in~$G$. \\

For a tree $T\in UB(n)$, the size of the first SPR neighbourhood is given by
$$|N_{SPR}(T)|=2(n-3)(2n-7).$$
This was determined by Allen and Steel~\cite{Steel}. No other SPR neighbourhood sizes are currently known. \\

In relation to the structure of the SPR neighbourhood, Carceres {\it et al.}~\cite{SPRWalks} provided tight bounds on the length of the shortest NNI walk that visits all trees in the first SPR neighbourhood of a tree~$T$. Allen and Steel~\cite{Steel} found upper and lower bounds for the maximum SPR distance between any two trees in~$UB(n)$. \\

As with NNI and RF, the size of the first SPR neighbourhood of a tree depends only on the number of leaves in the tree. However, unlike NNI and RF, the size of the second SPR neighbourhood of a tree cannot be expressed solely in terms of the number of leaves and cherries of the tree. In this section we show that these two parameters are not sufficient to determine even the highest order term of the size of the second SPR neighbourhood. At the end of this section we prove our main results, which are presented in Theorems~\ref{First} and~\ref{Second}. \\

\begin{theorem}
Let $T\in UB(n)$. 
\begin{enumerate}
\item[(i)] If $T$ is a caterpillar, then $$|N^2_{SPR}(T)|=\frac{1}{2}n^4+O(n^3).$$
\item[(ii)] If $T$ is a balanced tree, then $$|N^2_{SPR}(T)|=\frac{1}{3}n^4+O(n^3).$$
\end{enumerate}
\label{First}
\end{theorem}

It is evident from Theorem~\ref{First} that the size of the second SPR neighbourhood of a tree $T$ is not uniquely determined by the number of leaves of~$T$. However, every caterpillar has exactly two cherries, whereas a balanced tree with at least six leaves has at least three cherries. Therefore, for $n\geq 6$, a caterpillar and a balanced tree, each with $n$ leaves, have different numbers of cherries. Therefore Theorem~\ref{First} does not justify that the size of the second SPR neighbourhood of $T$ cannot be uniquely determined by the number of leaves and cherries of~$T$. To show this, we consider two different structures of an unrooted binary tree $T$ with $n=3m$ ($m\geq 3$) leaves and three cherries. These two tree structures (Type I and Type II) can be seen in Fig.~\ref{T1} and Fig.~\ref{T2} respectively. Similar to Theorem~\ref{First}, we show that trees of Type I and Type II also have a different highest order term in the expression for the size of the second SPR neighbourhood. This result is presented in Theorem~\ref{Second}. \\

\begin{figure}[htb]
\centering
\epsfig{file = 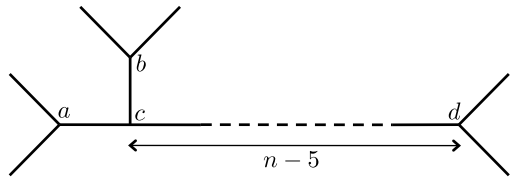, width = 0.5\linewidth}
\caption{A Type I tree with three cherries and $n=3m$ leaves ($m\geq 3$).}
\label{T1}
\end{figure}

\begin{figure}[htb]
\centering
\epsfig{file = 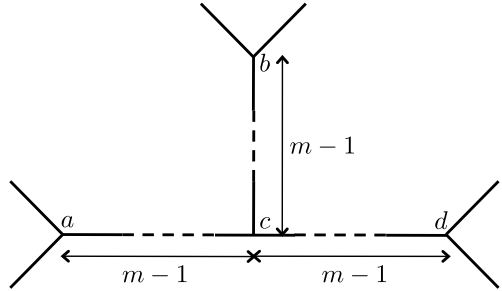, width = 0.5\linewidth}
\caption{A Type II tree with three cherries and $n=3m$ leaves ($m\geq 3$).}
\label{T2}
\end{figure}

\begin{theorem}
Let $T_1$ and $T_2$ be unrooted binary trees with $n=3m$ leaves ($m\geq 3$) and three cherries, and suppose that $T_1$ is of Type I and $T_2$ is of Type II. Then 
$$|N^2_{SPR}(T_1)|=\frac{1}{2}n^4+O(n^3), \text{ and}$$
$$|N^2_{SPR}(T_2)|=\frac{23}{54}n^4+O(n^3).$$
\label{Second}
\end{theorem}

We will use the notation 
$$SPR(T,(c_1,j_1),(c_2,j_2),...,(c_k,j_k))$$
to denote the tree obtained by $k$ successive SPR operations starting with tree $T$, where $c_1$ and $j_1$ in $T$ are the cut and join edges respectively, of the first operation, $c_2$ and $j_2$ in $SPR(T,(c_1,j_1))$ are the cut and join edges of the second operation, and so on. When $k=2$, we refer to the two operations that result in the set of trees $SPR(T,(c_1,j_1),(c_2,j_2))$  as a {\it pair of SPR operations}. It is worth noting that some of these cut and join edges may not be edges of $T$ if they are created by one of the SPR operations. However, the results in this section will require only sets of `well-separated' edges, where the cut and join edges are pairwise at least distance three apart, so that all of the edges $c_1,...,c_k$ and $j_1,...,j_k$ are edges of~$T$.\\

First, we determine an upper bound on the size of the second SPR neighbourhood. This follows directly from the expression for the size of the first SPR neighbourhood given by Allen and Steel~\cite{Steel}.\\

\begin{corollary}
Let $T\in UB(n)$ ($n\geq 3$). Then
$$|N^2_{SPR}(T)|\leq 4(n-3)^2(2n-7)^2=O(n^4).$$
\end{corollary}

The first step in proving Theorems~\ref{First} and~\ref{Second} is to determine whether or not all pairs of SPR operations contribute to the term of order $n^4$ in the expression for the size of the second SPR neighbourhood of a tree. \\

Let $T\in UB(n)$ and let
$$\mathbb{T}(T)=\{(c_1,c_2,j_1,j_2):c_1,j_1\in E(T), c_1\neq j_1; c_2,j_2\in E(SPR(T,(c_1,j_1))), c_2\neq j_2\}.$$
This is the set of all possible choices for the four cut and join edges of two SPR operations starting with tree~$T$. \\

Let $\mathbb{S}(T)$ be the subset of $\mathbb{T}(T)$ where $c_2, j_2\in E(T)$ and the four edges $c_1$, $j_1$, $c_2$, $j_2$ are pairwise at least distance three apart in~$T$. \\

The following lemma shows that in order to prove Theorems~\ref{First} and~\ref{Second}, it suffices to consider only pairs of SPR operations with cut and join edges in~$\mathbb{S}(T)$. \\

\begin{lemma}
Let $T\in UB(n)$. Then
$$|\mathbb{S}(T)|=\frac{2}{3}n^4+O(n^3)$$
$$|\mathbb{T}(T)-\mathbb{S}(T)|=O(n^3).$$
\label{size}
\end{lemma}
\begin{proof}
For sufficiently large values of $n$, it is possible to choose the edges $c_1$, $j_1$, $c_2$ and $j_2$ in $T$ such that $(c_1,c_2,j_1,j_2)\in\mathbb{S}(T)$. To determine the size of $\mathbb{S}(T)$, we count the number of sets of four internal edges of $T$, where all pairs of edges in the set are at least distance three apart. There are $2n-3$ choices for edge $c_1$, since this is the number of edges in $T$ (this follows from Lemma~\ref{intedges}). The maximum number of choices for $j_1$ is $(2n-3-7)$ (this can occur if $c_1$ is a pendant edge). The minimum number of choices for edge $j_1$ is $(2n-3-29)$ (this can occur if $c_1$ is an internal edge). The maximum number of choices for $c_2$ is $(2n-3-7-6)$ (this can occur if $c_1$ and $j_1$ are both pendant edges). The minimum number of choices for $c_2$ is $(2n-3-2(29))$ (this can occur if both $c_1$ and $j_1$ are internal edges). A similar process determines upper and lower bounds on the number of choices for edge~$j_2$. We divide by the number of ways to order the four edges. Therefore
$$|\mathbb{S}(T)|\geq \frac{1}{4!}(2n-3)(2n-3-29)(2n-3-2(29))(2n-3-3(29))=\frac{2}{3}n^4+O(n^3), \text{ and}$$
$$|\mathbb{S}(T)|\leq \frac{1}{4!}(2n-3)(2n-3-7)(2n-3-7-6)(2n-3-7-2(6))=\frac{2}{3}n^4+O(n^3).$$

We now consider $\mathbb{T}(T)-\mathbb{S}(T)$. Determining $|\mathbb{T}(T)-\mathbb{S}(T)|$ is similar to determining $|\mathbb{S}(T)|$, however for at least one of the four cut and join edges, instead of counting the number of edges at least distance three from those already chosen, we count the number within distance two of those already chosen, and therefore obtain a constant factor instead of a linear factor. Let $M$ be a maximal subset of the the edges $\{c_1,c_2,j_1,j_2\}$ such that the edges in $M$ are pairwise distance at least three apart in $T$, where $|M|=m<4$. Suppose we first choose the edges in~$M$. From the argument above we can see that the number of such choices is~$O(n^m)$. The remaining $4-m\geq 1$ edges must be chosen from edges within distance two of those already chosen. The number of these choices depends only on the number and location of the $m$ edges already chosen, and not on~$n$. Hence
$$|\mathbb{S}(T)|=\frac{2}{3}n^4+O(n^3)\text{, and   }|\mathbb{T}(T)-\mathbb{S}(T)|=O(n^3).$$
\end{proof}

Lemma~\ref{size} tells us that the highest order term in the expression for the size of $\mathbb{S}(T)$ is~$O(n^4)$. Note that instead of requiring the edges in $\mathbb{S}(T)$ to be at least distance three apart, we could have made them distance $k$ apart for any $k\in\mathbb{Z}^+$ and Lemma~\ref{size} would still hold. We have chosen to consider distance three, because if pairs of these four edges are within distance two of each other, then there are more cases to consider in order to determine exactly when two different pairs of SPR operations produce the same tree. To determine only the $O(n^4)$ term in the expression for the size of the second SPR  neighbourhood, we can ignore all cases where there exist edges $e,f\in \{c_1,c_2,j_1,j_2\}$ such that $d_T(e,f)\leq 2$. \\

However, we cannot simply take the highest order term in the expression for the size of $\mathbb{S}(T)$ as the highest order term in the expression for the size of the second SPR neighbourhood of a tree~$T$. This is because there may be cases where two different pairs of SPR operations produce the same tree (duplicates), or when a pair of SPR operations produces a first SPR neighbour of~$T$. To prove Theorem~\ref{First} and Theorem~\ref{Second} we need to know precisely when these two situations arise.  \\

In Lemma~\ref{notequal} and Lemma~\ref{equal} we show that a pair of SPR operations with cut and join edges in $\mathbb{S}(T)$ never yields a first SPR neighbour. We also consider when it is possible for two pairs of SPR operation to produce the same tree. We consider different cases for the relative locations of the four cut and join edges. There are three cases to consider:
\begin{enumerate}
\item The edges $j_2, c_1, c_2, j_1$ lie on a path in $T$ in this order.
\item The edges $c_2, c_1, j_2, j_1$ lie on a path in $T$ in this order.
\item If the four cut and join edges lie on a path in $T$, then they are in an order other than those given in Cases (1) and (2).
\end{enumerate}
Note that Cases (1) and (2) are characterised by the cut edge $c_1$, of the first operation, being on the path between the cut and join edges $c_2$ and $j_2$, of the second operation. In both of these cases, the first operation changes which internal vertices form the endpoints of the $(c_1,j_1)$-path. In Case (3), the first operation does not change these endpoints. In Lemma \ref{equal}, we see that in Cases (2) and (3), 
$$SPR(T,(c_1,j_1)(c_2,j_2))=SPR(T,(c_2,j_2)(c_1,j_1)).$$
However, in Lemma \ref{notequal}, we show that in Case (1), 
$$SPR(T,(c_1,j_1)(c_2,j_2))\neq SPR(T,(c_2,j_2)(c_1,j_1)).$$
For all three cases, we see that no other pair of SPR operations can yield the tree\\ \text{$SPR(T,(c_1,j_1)(c_2,j_2))$}. \\

First, we require a result about how SPR operations on a tree $T\in UB(n)$ with subtrees $A$ and $B$ can result in a tree $T'$ where $d_{T'}(A,B)<d_{T}(A,B)$. Recall that an internal subtree always has at least two vertices of degree two. Therefore, an internal subtree must have at least one internal edge.\\

\begin{lemma}
Let $T\in UB(n)$. Suppose there exist subtrees $A$ and $B$ of $T$ (not necessarily pendant or maximal), such that $d_T(A,B)= k$. Let $a$ and $b$ be vertices of degree two in $A$ and $B$ respectively, such that $d_T(a,b)=k$. Call the two pendant edges of the $(a-b)$-path $P$, $e$ and $f$ respectively. Let $T'\in UB(n)$ and suppose that $A$ and $B$ are subtrees of $T'$ such that $d_{T'}(A,B)=2$. Let $c_1$, $j_1$, $c_2$, and $j_2$ be edges of~$T$.
\begin{enumerate}
\item[(i)] For $k\geq 4$, if $T'=SPR(T, (c_1,j_1))$, then $c_1\in\{e,f\}$ and $j_1$ is incident to $A$ or $B$. 
\item[(ii)] For $k\geq 5$, if $T'=SPR(T, (c_1,j_1),(c_2,j_2))$ with $(c_1,c_2,j_1,j_2)\in\mathbb{S}(T)$, then either $c_1\in\{e,f\}$ and $j_1$ is incident to $A$ or $B$, or $c_2\in\{e,f\}$ and $j_2$ is incident to $A$ or $B$.
\item[(iii)] For $k\geq 4$, if $d_{T'}(a,b)=2$ and $T'=SPR(T, (c_1,j_1))$, then $\{c_1,j_1\}=\{e,f\}$.
\item[(iv)] For $k\geq 5$, if $d_{T'}(a,b)=2$ and $T'=SPR(T, (c_1,j_1),(c_2,j_2))$ with $(c_1,c_2,j_1,j_2)\in\mathbb{S}(T)$, then $\{c_1,j_1\}=\{e,f\}$ or $\{c_2,j_2\}=\{e,f\}$.\\
\end{enumerate}
\label{extra}
\end{lemma}
\begin{proof}

Fig.~\ref{ExtraSPR} shows trees $T$ and~$T'$. 
\begin{figure}[htb]
\centering
\epsfig{file = 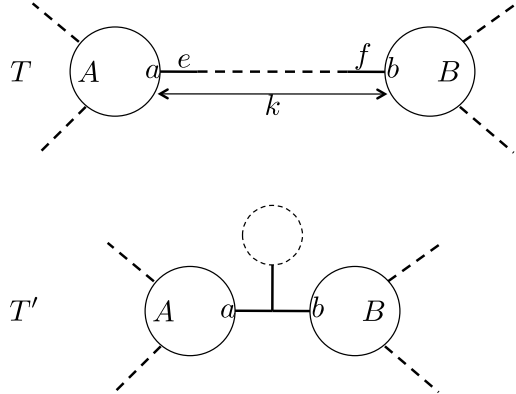, width = 0.5\linewidth}
\caption{Trees $T$ and~$T'$ from Lemma~\ref{extra}. Here we have $T'$ with $d_{T'}(a,b)=2$. }
\label{ExtraSPR}
\end{figure}
\begin{enumerate}
\item[(i)] Let $T''=SPR(T,(c_1,j_1))$. We show that $T''\neq T'$ unless $c_1\in\{e,f\}$ and $j_1$ is incident to $A$ or $B$. First suppose that edge $c_1$ is in subtree $A$ or~$B$. Then either $T''=T$ (if we regraft in the same place), or $T''\neq T'$ since the subtree ($A$ or $B$) is not in~$T''$. Likewise for edge~$j_1$. \\

Now suppose that $c_1$ and $j_1$ are not edges of $A$ or~$B$. Hence $A$ and $B$ are both subtrees of~$T''$. If we assume that edge $c_1$ is not in $P$ or incident to $P$, then $d_{T''}(A,B)\geq 5$ if $j_1$ is an edge of $P$; otherwise, $d_{T''}(A,B)\geq 4$. Therefore~$T''\neq T'$. If $c_1$ is incident to $P$ then deleting edge $c_1$ creates a vertex of degree two in $P$, which is suppressed by the SPR operation. Hence $d_{T''}(A,B)\geq 4$ if $j_1$ is an edge of $P$; otherwise, $d_{T''}(A,B)\geq 3$. Therefore~$T''\neq T'$. \\ 

Now suppose that $j_1$ is not an edge of $A$ or $B$, and $c_1$ is an edge of~$P$. If $c_1\not\in\{e,f\}$ then $d_{T''}(A,B)\geq 3$ if $j_1$ is incident to $A$ or $B$; otherwise, $d_{T''}(A,B)\geq 4$.  Therefore~$T''\neq T'$. Finally suppose that $c_1\in\{e,f\}$, and $j_1$ is not incident to $A$ or~$B$. Then $d_{T''}(A,B)\geq 3$ and~$T''\neq T'$. The only remaining possibility is that $c_1\in\{e,f\}$ and $j_1$ is incident to either $A$ or~$B$. \\

\item[(ii)] Let $T''=SPR(T,(c_1,j_1)(c_2,j_2))$, where $(c_1,c_2,j_1,j_2)\in\mathbb{S}(T)$. We show that $T''\neq T'$ unless either $c_1\in\{e,f\}$ and $j_1$ is incident to $A$ or $B$, or $c_2\in\{e,f\}$ and $j_2$ is incident to $A$ or $B$. As in $(i)$ if any of the four cut and join edges are in the subtrees $A$ or $B$ in $T$, then that subtree is not a subtree of $T''$, so~$T''\neq T'$. In $(i)$ we saw that if the cut edge of an operation is not in $P$ or incident to $P$ in $T$, then the operation does not reduce the distance between $A$ and~$B$. As in $(i)$, an operation with a cut edge incident to $P$, reduces the distance between $A$ and $B$ by at most one. Hence if neither cut edge $c_1$ nor $c_2$ is in $P$, we have $d_{T''}(A,B)\geq 3$, and so~$T''\neq T'$. \\

Now we assume that at least one of the edges $c_1$ and $c_2$ is an edge of~$P$. Suppose that $c_1$ is not an edge of $P$, but $c_2$ is. Then if $T_1=SPR(T,(c_1,j_1))$, $d_{T_1}(A,B)\geq 4$. By $(i)$, if $T''=T'$ then $c_2\in\{e,f\}$ and $j_2$ is incident to either $A$ or $B$.\\

Now suppose that $c_1\not\in\{e,f\}$ is an edge of~$P$. Then as in $(i)$, $d_{T_1}(A,B)\geq 3$ if $j_1$ is incident to $A$ or $B$, and $d_{T_1}(A,B)\geq 4$ otherwise. If \text{$d_{T_1}(A,B)\geq 4$} then by $(i)$, unless \text{$c_2\in\{e,f\}$} and $j_2$ is incident to $A$ or $B$, the second operation cannot result in $T'$. If $d_{T_1}(A,B)=3$, then since $(c_1,c_2,j_1,j_2)\in\mathbb{S}(T)$, the edges $c_2$ and $j_2$ cannot be in or incident to the shortest path between $A$ and $B$ in~$T_1$. Hence $d_{T''}(A,B)=3$, and~$T''\neq T'$. \\

Finally, suppose that $c_1\in\{e,f\}$. If $j_1$ is not incident to $A$ or $B$ then in the tree $T_1=SPR(T,(c_1,j_1))$ we have $d_{T_1}(A,B)\geq 3$. Again, if $d_{T_1}(A,B)\geq 4$ then by $(i)$, the second operation cannot result in $T'$ unless $c_2\in\{e,f\}$ and $j_2$ is incident to $A$ or $B$. If $d_{T_1}(A,B)=3$ then since $(c_1,c_2,j_1,j_2)\in\mathbb{S}(T)$, the edges $c_2$ and $j_2$ cannot be in or incident to the shortest path between $A$ and $B$ in~$T_1$. Hence $d_{T''}(A,B)=3$, and~$T''\neq T'$. The only remaining possibility is that $j_1$ is incident to either $A$ or~$B$. \\

\item[(iii)] Now we have $d_{T'}(a,b)=2$. Let $T''=SPR(T,(c_1,j_1))$. By $(i)$, if $T''=T'$ then $c_1\in\{e,f\}$ and $j_1$ is incident to either $A$ or~$B$. If $j_1\not\in\{e,f\}$ (which can occur if $A$ or $B$ is an internal subtree) then $d_{T''}(A,B)=2$ but $d_{T''}(a,b)>2$, since the internal subtree has at least one internal edge. Hence if $T''=T'$ then $\{c_1,j_1\}=\{e,f\}$. \\

\item[(iv)] Let $T''=SPR(T,(c_1,j_1)(c_2,j_2))$, where $(c_1,c_2,j_1,j_2)\in\mathbb{S}(T)$. By $(ii)$, if $T''=T'$, then either $c_1\in\{e,f\}$ and $j_1$ is incident to $A$ or $B$, or $c_2\in\{e,f\}$ and $j_2$ is incident to $A$ or $B$. \\

Suppose that $c_1\in\{e,f\}$ and $j_1$ is incident to $A$ or~$B$. Let $T_1=SPR(T,(c_1,j_1))$. If $j_1\not\in\{e,f\}$, then $d_{T_1}(A,B)=2$, but $d_{T_1}(a,b)>2$. This is because $a$ and $b$ are not the endpoints of the shortest path between $A$ and $B$ in~$T_1$. Now, $c_2$ and $j_2$ cannot be edges in $A$ or $B$, and $(c_1,c_2,j_1,j_2)\in\mathbb{S}(T)$. Therefore, $c_2$ and $j_2$ cannot be edges on or incident to the path between $a$ and $b$ in~$T_1$. Therefore $d_{T''}(a,b)>2$, and $T''\neq T'$. Hence, we must have $\{c_1,j_1\}=\{e,f\}$. \\

Now suppose that $c_2\in\{e,f\}$, and $j_2$ is incident to $A$ or $B$. If $c_1$ is not an edge of $P$, then as in $(i)$, $d_{T_1}(A,B)\geq 4$. Therefore, by $(iii)$, if $T''=T'$ then $\{c_2,j_2\}=\{e,f\}$. Now suppose that $c_1$ is an edge of $P$. If $a$ and $b$ are the endpoints of the shortest path between $A$ and $B$ in $T_1$, then $d_{T_1}(A,B)\geq 6$, and so by $(iii)$, if $T''=T'$ then $\{c_2,j_2\}=\{e,f\}$. Now assume that $a$ and $b$ are not the endpoints of the shortest path $P'$ between $A$ and $B$ in~$T_1$. Then exactly one of the edges $e$ or $f$ is an edge of~$P'$. Without loss of generality, suppose that $e$ is an edge in~$P'$. Then if $T''=T'$, $c_2=e$ and $j_2$ is incident to~$B$. If $j_2\neq f$, then we have $d_{T''}(A,B)=2$, but $d_{T''}(a,b)>2$. Therefore $T''\neq T'$. Hence if $T''=T'$, then $\{c_2,j_2\}=\{e,f\}$. 
\end{enumerate}
\end{proof}

\begin{lemma}
Let $T\in UB(n)$, and suppose that we have trees $T'=SPR(T,(c_1,j_1))$ and \text{$T''=SPR(T,(c_1,j_1),(c_2,j_2))$} where $(c_1, c_2, j_1, j_2)\in \mathbb{S}(T).$ Suppose that the edges $j_2$, $c_1$, $c_2$, and $j_1$ lie on a path in $T$ in this order. Then
\label{notequal}
\begin{enumerate}
\item[(i)] $T''\not\in N_{SPR}(T)$, and
\item[(ii)] for all other choices of edges $(c_1', c_2', j_1', j_2')\in \mathbb{S}(T)$ where $(c_1', c_2', j_1', j_2')\neq (c_1, c_2, j_1, j_2)$, we have
$$T''\neq SPR(T,(c_1',j_1'),(c_2',j_2')).$$
\end{enumerate}
\end{lemma}
\begin{proof}
Since the four cut and join edges lie on a path in $T$, the rest of the tree can be partitioned into five subtrees (two pendant and three internal) connected by these four edges.\\

Consider the forest $T\setminus\{c_1,j_1,c_2,j_2\}$. It has components $A$, $B$, $C$, $D$, and $E$ which are subtrees of~$T$. Edge $j_2$ is incident to $A$ and $B$, edge $c_1$ is incident to $B$ and $C$, edge $c_2$ is incident to $C$ and $D$, and edge $j_1$ is incident to $D$ and~$E$. Fig.~\ref{PC} shows $T$, $T'$ and~$T''$. Each of the internal subtrees $B$, $C$ and $D$ have at least three internal edges, as all pairs of the four cut and join edges are at least distance three apart. Let $b$ be the endpoint of $c_1$ that is in $B$, and $c$ be the endpoint of $c_2$ that is in~$C$.\\

\begin{figure}[htb]
\centering
\epsfig{file = 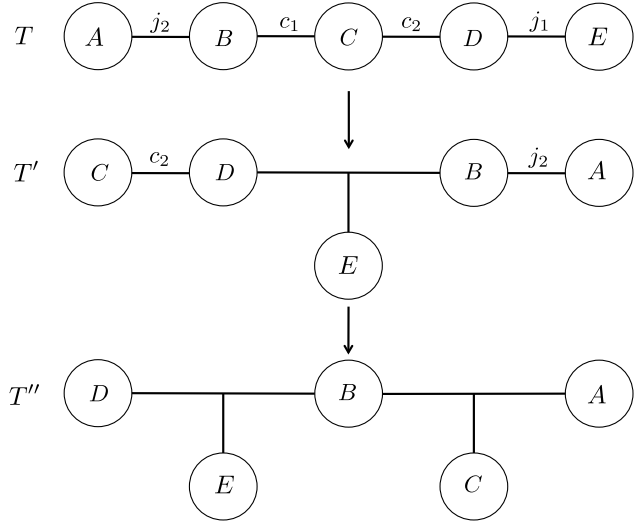, width = 0.7\linewidth}
\caption{Tree $T'=SPR(T,(c_1,j_1))$ and $T''=SPR(T,(c_1,j_1),(c_2,j_2))$.}
\label{PC}
\end{figure}

\begin{enumerate}
\item[(i)] From the above, we have $d_T(B,E)=d_T(b,E)\geq 9$ and $d_{T''}(B,E)=d_{T''}(b,E)=2$. Therefore if $T''$ is a first SPR neighbour of $T$, then either $T''=SPR(T,(c_1,j_1))=T'$ or $T''=SPR(T,(j_1,c_1))$ by Lemma~\ref{extra} $(iii)$ \footnote{Note that Lemma~\ref{extra} applies when $d_T(B,E)=d_T(b,x)\geq 9$ and $d_{T''}(B,E)=d_{T''}(b,x)=2$ where $x$ is a vertex of degree two in~$E$. However, since $E$ is a pendant subtree with only one vertex of degree two, we simply use $d_T(b,E)$ instead of $d_T(b,x)$ for simplicity. This occurs in other places throughout the proofs of Lemma~\ref{notequal} and Lemma~\ref{equal}.}. We have $d_{T''}(A,C)=2$ and $d_{T'}(A,C)\geq 10$, so~$T''\neq T'$. In $T_1=SPR(T,(j_1,c_1))$, $d_{T_1}(A,C)\geq 6$, so~$T''\neq T_1$. Therefore $T''$ is not a first SPR neighbour of~$T$.\\

\item[(ii)] Considering $(c_1',c_2',j_1',j_2')\in\mathbb{S}(T)$, suppose that we have $T_1=SPR(T,(c_1',j_1'))$ and \text{$T_2=SPR(T,(c_1',j_1'),(c_2',j_2'))$}. We show that if $T_2=T''$, \text{$(c_1', c_2', j_1', j_2')= (c_1, c_2, j_1, j_2)$}. As before, we have $d_{T''}(B,E)=d_{T''}(b,E)=2$. Since $d_T(B,E)=d_T(b,E)\geq 9$, if $T''=T_2$, then the cut and join edges for one of the operations must be $c_1$ and $j_1$ by Lemma~\ref{extra} $(iv)$. There are four cases to consider. \\

\begin{enumerate}
\item First suppose that $(c_1',j_1')=(c_1,j_1)$. Then~$T_1=T'$. Since all SPR operations on $T'$ result in distinct neighbours (see~\cite{Steel}), $T''=SPR(T,(c_1,j_1),(c_2',j_2'))$ only if $(c_2',j_2')=(c_2,j_2)$. \\

\item Now suppose that $(c_1',j_1')=(j_1,c_1)$. Then $d_{T_1}(B,C)=d_{T_1}(B,E)=d_{T_1}(C,E)=2$. The edges $c_2'$ and $j_2'$ must be distance three or more from $c_1$ and $j_1$ in~$T$. If $c_2'$ is in one of the subtrees $B$, $C$ and $E$, then this subtree is not a subtree of~$T_2$, and~\text{$T_2\neq T''$}. Similarly, if $j_2'$ is in one of these three subtrees, then~\text{$T_2\neq T''$}. If neither $c_2'$ or $j_2'$ are in one of the subtrees $B$, $E$ or $C$, then we know that \text{$d_{T_2}(B,C)=d_{T_2}(B,E)=d_{T_2}(C,E)=2$.} However, $d_{T''}(C,E)\geq 7$ so~$T_2\neq T''$. \\

\item We now assume that $\{c_2',j_2'\}=\{c_1,j_1\}$. Therefore $c_1',j_1'\not\in\{c_1,j_1\}$. We have $d_T(A,C)\geq 5$ and $d_{T''}(A,C)=2$. If $T_2=T''$, then by Lemma~\ref{extra} $(ii)$, we have $(c_1',j_1')=(j_2,c_2)$. Therefore, $d_{T_1}(A,C)=d_{T_1}(A,D)=2$. Regardless of whether the second SPR operation involves pruning $B$ or $E$ in $T_1$, \text{$d_{T_2}(A,C)=d_{T_2}(A,D)=2$.} However, $d_{T''}(A,D)\geq 7$, so $T_2\neq T''$. \\

\end{enumerate}

Therefore, $T_2=T''$ implies that $(c_1', c_2', j_1', j_2')= (c_1, c_2, j_1, j_2)$.\\
\end{enumerate}
\end{proof}

\newpage
\begin{lemma}
Let $T\in UB(n)$ and suppose that we have trees $T'=SPR(T,(c_1,j_1))$ and \text{$T''=SPR(T,(c_1,j_1),(c_2,j_2))$} where $(c_1, c_2, j_1, j_2)\in \mathbb{S}(T).$ Suppose that there is no path in $T$ in which the edges $j_2$, $c_1$, $c_2$, and $j_1$ lie in this order. Then
\label{equal}
\begin{enumerate}
\item[(i)] $T''\not\in N_{SPR}(T)$, and
\item[(ii)] for all choices of edges $(c_1', c_2', j_1', j_2')\in \mathbb{S}(T)$, $(c_1', c_2', j_1', j_2')\neq (c_1, c_2, j_1, j_2)$, we have
$$T''= SPR(T,(c_1',j_1'),(c_2',j_2'))$$
if and only if $(c_1', c_2', j_1', j_2')=(c_2,c_1,j_2,j_1)$. \\
\end{enumerate}
\end{lemma}
\begin{proof}
Let the endpoints of the $(c_1-j_1)$-path in $T$ be $c$ and $d$ respectively. Let the subtrees at the endpoints of the $(c_1-j_1)$-path be $C_1$ and $D_1$ respectively. Then $d_{T'}(C_1,D_1)=2$. Now let $C$ and $D$ be subtrees of $C_1$ and $D_1$ respectively for which $d_{T''}(C,D)=2$. Because neither $c_2$ nor $j_2$ is within distance two of $c_1$ or $j_1$, $C$ and $D$ each have at least three internal edges. Therefore, $C$ and $D$ are subtrees such that $d_T(C,D)=d_T(c,d)\geq 5$ and $d_{T''}(C,D)=d_{T''}(c,d)=2$ (note that $C$ and $D$ may be internal subtrees). \\

Suppose that the edges $c_2, c_1, j_2, j_1$ do not lie on a path in $T$ in this order (Case (3) from before the statement of Lemma~\ref{extra}). Let the endpoints of the $(c_2-j_2)$-path in $T$ be $a$ and $b$ respectively. Let the subtrees at the endpoints of the $(c_2-j_2)$-path be $A_1$ and $B_1$ respectively.  Let the subtrees at the endpoints of the $(c_2-j_2)$-path in $T'$ be $A_2$ and $B_2$ respectively. Note that $a$ and $b$ are also the endpoints of this path in~$T'$. Now let $A=A_1\cap A_2$ and $B=B_1\cap B_2$. Since $(c_1, c_2, j_1, j_2)\in \mathbb{S}(T)$, $A$ and $B$ have at least three internal edges. Note that $A$ and $B$ are subtrees at the endpoints of the $(c_2-j_2)$-path in~$T$ respectively ($A$ and $B$ may be internal subtrees of~$T$). We have $d_T(A,B)=d_T(a,b)\geq 5$. Since $c_1$ cannot be within distance two of either $c_2$ or $j_2$, $d_{T'}(A,B)=d_{T'}(a,b)\geq 5$. Finally, $d_{T''}(A,B)=d_{T''}(a,b)=2$. \\

Now suppose that the edges $c_2, c_1, j_2, j_1$ lie on a path in $T$ in this order (Case (2) from before the statement of Lemma~\ref{extra}). Let $A$ be the pendant subtree of $T$ incident to $c_2$, and let $B$ be the internal subtree of $T$ incident to both $j_1$ and $j_2$. Let the endpoints of the $(c_2-j_2)$-path in $T$ be $a$ and $b$ respectively. Then in $T$, $A$ and $B$ are subtrees at the endpoints of the $(c_2-j_2)$-path respectively. Figure \ref{New} shows subtrees $A$, $B$, $C$ and $D$. Note that $d_{T}(A,B)=d_{T}(a,b)\geq 9$, and $d_{T''}(A,B)=d_{T''}(a,b)=2$. The difference between this case and Case (3), is that here, $a$ and $b$ are not the endpoints of the shortest path between $A$ and $B$ in $T'$, as they are in Case (3). \\

\begin{figure}[htb]
\centering
\epsfig{file = 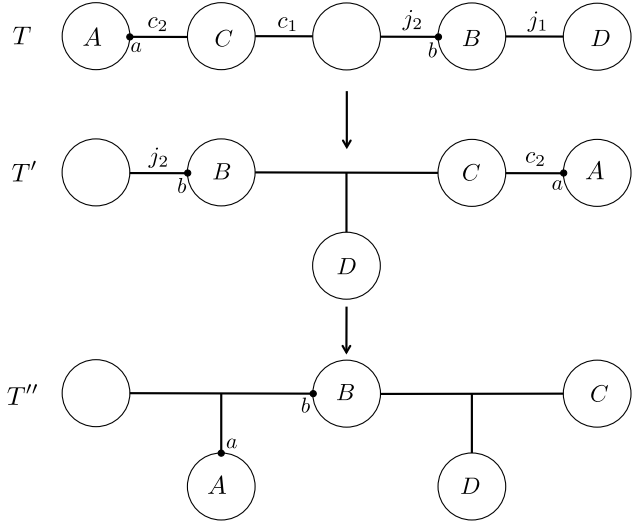, width = 0.7\linewidth}
\caption{Trees $T$, $T'=SPR(T,(c_1,j_1))$ and $T''=SPR(T,(c_1,j_1),(c_2,j_2))$ where the edges $c_2, c_1, j_2, j_1$ lie on a path in $T$ in this order.}
\label{New}
\end{figure}

\begin{enumerate}
\item[(i)] From above, we have $d_T(C,D)=d_T(c,d)\geq 5$, but $d_{T''}(C,D)=d_{T''}(c,d)=2$. If $T''\in N_{SPR}(T)$, we either have $T''=SPR(T,(c_1,j_1))=T'$ or \text{$T''=SPR(T,(j_1,c_1))$} by Lemma~\ref{extra} $(iii)$. Now $d_{T''}(A,B)=2$, while $d_{T'}(A,B)\geq 5$, and so $T''\neq T'$. In \text{$T_1=SPR(T,(j_1,c_1))$}, $d_{T_1}(A,B)\geq 5$ because $c_1$ cannot be within distance two of either $c_2$ or $j_2$. Hence $T''\neq T_1$.\\

\item[(ii)] As in Lemma~\ref{notequal}, if $$T''= SPR(T,(c_1',j_1'),(c_2',j_2'))$$ where $(c_1',c_2',j_1',j_2')\in\mathbb{S}(T)$, and $(c_1',c_2',j_1',j_2')\neq (c_1,c_2,j_1,j_2)$, then by Lemma~\ref{extra} $(iv)$, we have \text{$\{\{c_1',j_1'\},\{c_2',j_2'\}\}=\{\{c_1,j_1\},\{c_2,j_2\}\}$}. We consider all possible cases. Let $T_1=SPR(T,(c_1',j_1'))$ and $T_2=SPR(T,(c_1',j_1'),(c_2',j_2'))$. \\

\begin{enumerate}
\item First let $(c_1',j_1')=(c_2,j_2)$ and $(c_2',j_2')=(c_1,j_1)$. In Case (3) the first SPR operation on $T$ prunes and regrafts $A_1$ so that $d_{T_1}(A_1,B_1)=2$. In Case (2) the first SPR operation on $T$ prunes and regrafts $A$ so that $d_{T_1}(A,B)=2$. In both cases, the endpoints of the $(c_1-j_1)$-path in $T_1$ are $c$ and~$d$. Hence $d_{T_2}(A,B)=d_{T_2}(a,b)=d_{T_2}(C,D)=d_{T_2}(c,d)=2$ and case analysis shows that~$T_2=T''$. So
$$T''=SPR(T,(c_1,j_1),(c_2,j_2))=SPR(T,(c_2,j_2),(c_1,j_1)).$$

\item Now consider the case where $(c_1',j_1')=(c_1,j_1)$. Then~$T_1=T'$. Since we know that SPR operations on $T$ with different cut and join edges result in distinct trees (see~\cite{Steel}), we have
$$SPR(T,(c_1,j_1),(j_2,c_2))\neq SPR(T,(c_1,j_1),(c_2,j_2))=T''.$$
Similarly,
$$SPR(T,(c_2,j_2),(j_1,c_1))\neq SPR(T,(c_2,j_2),(c_1,j_1))=T''.$$

\item Let $X$ be the subtree of $T$ such that $d_T(X,D)=2$ and $X$ does not contain edge~$c_1$. Then $d_{T'}(C,D)=2$ and $d_{T'}(C,X)=d_{T'}(D,X)=3$. Since the cut and join edges for the second SPR operation must be at least distance three from $c_1$ and $j_1$ in $T$, there is a subtree of $X$, which we denote $X'$, such that $X'$ contains the root of $X$, and \text{$d_{T''}(C,X')=d_{T''}(D,X')=3$}. Suppose that $(c_1',j_1')=(j_1,c_1)$. Then we know that \text{$d_{T_1}(C,X)=d_{T_1}(D,X)\geq 6$.} Again, there exists a subtree $X''$ of $X$ such that $X''$ contains the root of $X$, and $d_{T_2}(C,X'')=d_{T_2}(D,X'')\geq 6$. Since $(c_1, c_2, j_1, j_2)\in \mathbb{S}(T)$ and $(c_1', c_2', j_1', j_2')\in \mathbb{S}(T)$, the intersection $\inte{X}$ between $X'$ and $X''$ is non-empty, and $\inte{X}$ must have at least one leaf of $T$. Therefore~$T_2\neq T''$. The same argument applies with subtrees $A$ and $B$ if we consider $(c_1',j_1')=(j_2,c_2)$.\\

\end{enumerate}

Therefore
$$T''=SPR(T,(c_1,j_1),(c_2,j_2))=SPR(T,(c_2,j_2),(c_1,j_1)),$$
but for all other choices of edges $(c_1', c_2', j_1', j_2')\in \mathbb{S}(T)$, $(c_1', c_2', j_1', j_2')\neq (c_1, c_2, j_1, j_2)$, we have
$T''\neq SPR(T,(c_1',j_1'),(c_2',j_2'))$.\\

\end{enumerate}
\end{proof}

We have now established that every pair of SPR operations on a tree $T\in UB(n)$ produces a second SPR neighbour $T$ (not a first SPR neighbour). The only case where two different pairs of SPR operations produce the same second neighbour arises when there is no path in $T$ with the edges $j_2$, $c_1$, $c_2$, $j_1$ in the order listed, and
$$SPR(T,(c_1,j_1),(c_2,j_2))=SPR(T,(c_2,j_2),(c_1,j_1)).$$ 

We now count the number of ways the edges $j_2$, $c_1$, $c_2$ and $j_1$ can appear in a path in a binary tree $T$ in the order given, where the four edges are pairwise at least distance three apart (i.e. they are in $\mathbb{S}(T)$). Let this quantity be~$P(T)$. In order to calculate $P(T)$, we need to determine the number of paths of all lengths greater than or equal to $13$ in $T$. This quantity is not uniquely determined by the number of leaves and cherries of $T$ (Theorem~\ref{pathsl3}). However, if $T$ is known to be a caterpillar or a balanced tree, we can determine the number of paths in $T$ of any given length, using the number of leaves of~$T$. In Lemma \ref{pathcount} we find expressions for the number of paths of a given length in a caterpillar and a balance tree, which we will use to prove Theorem~\ref{First}.\\

\begin{lemma}\mbox{}
For $n\geq 4$:
\begin{enumerate}
\item[(i)] A caterpillar with $n$ leaves has $4(n-k)$ paths of length $k$ for $3\leq k\leq n-1$.
\item[(ii)]
Let 
\begin{displaymath}
   f(k) = \left\{
     \begin{array}{lr}
       3\left(2^{\frac{k}{2}-1}\right)\left(n-2^{\frac{k}{2}}\right), & k \text{ even};\\
       2^{\frac{k+1}{2}}\left(n-3\left(2^{\frac{k-3}{2}}\right)\right), & k \text{ odd}.
     \end{array}
   \right.
\end{displaymath}
A balanced tree with $n=2^i$ leaves ($i\geq 2$) has $f(k)$ paths of length $k$ for \text{$3\leq k\leq 2i-1$}, and a balanced tree with $n=3\cdot 2^i$ leaves ($i\geq 1$) has $f(k)$ paths of length $k$ for \text{$3\leq k\leq 2(i+1)$}.\\
\end{enumerate}
\label{pathcount}
\end{lemma}
\begin{proof}\mbox{}
\begin{enumerate}
\item[(i)] A caterpillar $T$ has a single path of $n-3$ internal edges. Now $p_{k-2}(T)$ is the number of ways to choose $k-2$ of these internal edges so that they are adjacent. This is given by $p_{k-2}(T)=(n-3)-(k-2)+1=n-k$. Then by Lemma~\ref{ps}, $P_k(T)=p_{k-2}(T)=4(n-k)$, for $k\geq 3$.\\
 
\item[(ii)] If $T$ is a balanced tree with $n$ leaves, then it has $c=\frac{n}{2}$ cherries. Let $\bar{P}_k(n)$ be the number of paths of length $k$ in a balanced tree with $n$ leaves, and let $\bar{p}_k(n)$ be the number of internal paths of length $k$ in a balanced tree with $n$ leaves. The number of internal paths of length $k$ in $T$ is given by the number of paths of length $k$ in $T'$ where $T'$ is the subtree induced by the internal vertices of~$T$. Since $T'$ has $\frac{n}{2}$ leaves, 
$$\bar{p}_k(n)=\bar{P}_k\left(\frac{n}{2}\right),$$
provided $n\geq 6$. From Lemma~\ref{ps}, $$\bar{P}_k(n)=4\bar{p}_{k-2}(n).$$  
We have $\bar{p}_2(n)=n+c-6=3\left(\frac{n}{2}-2\right)$ by Theorem~\ref{pathsl3}, so if $k$ is even then
\begin{align*}
\bar{P}_k(n)&=3\left(2^{k-2}\right)\left(\frac{n}{2^{\frac{k}{2}-1}}-2\right)\\
&=3\left(2^{\frac{k}{2}-1}\right)\left(n-2^{\frac{k}{2}}\right).
\end{align*}

We have $\bar{p}_1(n)=n-3$ by Lemma~\ref{intedges}, so if $k$ is odd then
\begin{align*}
\bar{P}_k(n)&=2^{k-1}\left(\frac{n}{2^{\frac{k-1}{2}-1}}-3\right)\\
&=2^{\frac{k+1}{2}}\left(n-3\left(2^{\frac{k-3}{2}}\right)\right).
\end{align*}

\end{enumerate}
Now if $n=2^i$, the maximum path length in the tree is given by $2i-1$, and if $n=3\cdot 2^i$ then the maximum path length in the tree is given by~$2(i+1)$. \\
\end{proof}

Now that we know the number of paths of any given length in a caterpillar or balanced tree, we can determine the size of~$P(T)$. We are now ready to to prove Theorem~\ref{First}.

\subsection*{Proof of Theorem \ref{First}}

Suppose that $T$ has a path $P$ of length $k$, $k\geq 13$. Fix the two pendant edges of $P$ as $j_2$ and $j_1$ so that $j_2$ is the first edge in $P$, and $j_1$ is the $k^{th}$ edge in~$P$. All pairs of the edges $j_2$, $c_1$, $c_2$, and $j_1$ must be distance three or more apart and in the order given. So $d_T(c_1,j_2)\geq 3$ and $d_T(c_1,j_1)\geq 7$. If $c_1$ is the $m^{th}$ edge in $P$ then $5\leq m\leq k-8$. Now if $c_2$ is the $j^{th}$ edge in $P$, then $m+4\leq j\leq k-4$, so there are $(k-4)-(m+4)+1=k-m-7$ possible choices for the location of~$c_2$. Finally, it does not matter at which endpoint of $P$ we begin counting. So the number of ways of arranging the four edges on this path is
$$R_k=2\sum_{m=5}^{k-8}(k-m-7)=(k-11)(k-12).$$
\begin{enumerate}
\item[(i)] By Lemma~\ref{pathcount}, $T$ has $4(n-k)$ paths of length $k$ for $k\geq 3$. Hence for a caterpillar, 
\begin{align}
P(T)&=\sum_{k=13}^{n-1}4(n-k)(k-11)(k-12)\\
&=\frac{1}{3}n^4+O(n^3).
\label{cater}
\end{align}
We know by Lemma~\ref{notequal} and Lemma~\ref{equal} that if we count the number of ways to choose the edges $(c_1, c_2, j_1, j_2)\in \mathbb{S}(T)$, then in the cases not counted by $P(T)$ we count every second neighbour twice. For the cases that are counted by $P(T)$ we do not obtain any duplicate trees. Therefore by Lemma~\ref{size},
\begin{align*}
|N^2_{SPR}(T)|&=\frac{1}{2}\left(\frac{2}{3}n^4+O(n^3)-P(T)\right)+P(T)\\
&=\frac{1}{2}\left(\frac{2}{3}n^4+P(T)\right)+O(n^3)\\
&=\frac{1}{2}\left(\frac{2}{3}n^4+\frac{1}{3}n^4\right)+O(n^3)=\frac{1}{2}n^4+O(n^3).
\end{align*}
\item[(ii)] Similarly for a balanced tree $T$ with $n=3(2)^i$ leaves ($i\geq 1$), we can sum over even and odd path lengths (see Lemma~\ref{pathcount}) to obtain
\begin{align*}
P(T)&=\sum_{k=13}^{n-1} P_k(T)(k-11)(k-12)\\
&=\sum_{m=7}^{\log_2(\frac{n}{3})+1}\left(3\left(2^{m-1}\right)\left(n-2^m\right)(2m-11)(2m-12) \right)+\\&\sum_{m=7}^{\log_2(\frac{n}{3})+1}\left(2^m\left(n-3\left(2^{m-2}\right)\right)(2m-12)(2m-13)\right)\\
&=\frac{8}{\ln(2)^2}n^2\ln(n)^2+O(n^2\ln(n))\\
&=O(n^2\ln(n)^2)=O(n^3).\\
\end{align*}
If $T$ is a balanced tree with $n=2^i$ leaves ($i\geq 2$), then instead we have
\begin{align*}
P(T)&=\sum_{m=7}^{\log_2(\frac{n}{4})+1}\left(3\left(2^{m-1}\right)\left(n-2^m\right)(2m-11)(2m-12) \right)+\\&\sum_{m=7}^{\log_2(\frac{n}{4})+2}\left(2^m\left(n-3\left(2^{m-2}\right)\right)(2m-12)(2m-13)\right)\\
&=\frac{8}{\ln(2)^2}n^2\ln(n)^2+O(n^2\ln(n))=O(n^3).\\
\end{align*}
Therefore, for any balanced tree $T$, $$|N^2_{SPR}(T)|=\frac{1}{2}\left(\frac{2}{3}n^4+P(T)\right)+O(n^3)=\frac{1}{3}n^4+O(n^3).$$
\end{enumerate}
\tallqed

This shows that the size of the second SPR neighbourhood of a tree cannot be uniquely determined by the number of leaves of the tree. We now prove Theorem \ref{Second}, which shows that the number of leaves and cherries is insufficient. 

\subsection*{Proof of Theorem \ref{Second}}

Suppose that $n=3m$ and $c=3$, where $m\geq 7$. Consider the tree $T_1$ of Type I, with $n$ leaves and $c$ cherries (see Fig.~\ref{T1}). For any pair of vertices $x,y$, let $C_{xy}$ be the caterpillar formed by the path between vertices $x$ and $y$ in $T_1$ and all of the edges incident to vertices on that path. Let $a$, $b$ and $d$ be the roots of the three cherries of $T_1$, such that $d_{T_1}(a,b)=2$. Let $c$ be the vertex in $T_1$ that is not adjacent to a leaf. Both of the caterpillars $C_{ad}$ and $C_{bd}$ have $n-1$ leaves. If we find $P(C_{ad})$ and $P(C_{bd})$, then we will have found every way of choosing the edges $c_1$, $c_2$, $j_1$ and $j_2$ so that all four edges are on a path in the order \text{$j_2$, $c_1$, $c_2$,~$j_1$}. Eliminating double counting, we have $$P(T_1)=P(C_{ad})+P(C_{bd})-P(C_{cd})=2P(C_{ad})-P(C_{cd}).$$
We do not consider the caterpillar $C_{ab}$ because it is too short to have any paths of length $13$ or more. So by Equation~\ref{cater},\\
$$P(T_1)=\frac{2}{3}(n-1)^4-\frac{1}{3}(n-2)^4+O(n^3)=\frac{1}{3}n^4+O(n^3).$$

Now let $T_2$ be the tree of Type II with $n$ leaves, $c$ cherries and maximum path length $2m$ (see Fig.~\ref{T2}). Let $a$, $b$ and $d$ be the roots of the three cherries of $T_2$, and let $c$ be the vertex in $T_2$ that is not adjacent to a leaf. By the same process as above,
$$P(T_2)=P(C_{ad})+P(C_{bd})+P(C_{ab})-P(C_{ac})-P(C_{bc})-P(C_{cd})=3P(C_{ad})-3P(C_{ac}).$$
Now $C_{ad}$ has $2m+1$ leaves and $C_{ac}$ has $m+2$ leaves, so
\begin{align*}
P(T_2)&=(2m+1)^4-(m+2)^4+O(n^3)\\
&=(\frac{2}{3}n+1)^4-(\frac{1}{3}n+2)^4+O(n^3)\\
&=\frac{5}{27}n^4+O(n^3).
\end{align*}

Therefore $|N^2_{SPR}(T_1)|=\frac{1}{2}n^4+O(n^3)$
and $|N^2_{SPR}(T_2)|=\frac{23}{54}n^4+O(n^3).$\\
\tallqed

Since $T_1$ and $T_2$ have the same number of leaves and cherries, it is clear that other properties of the tree $T$ would be required to get an exact formula for the highest order term of $|N^2_{SPR}(T)|$.

\section{Concluding Comments}

In this paper, we derived new results for the sizes of the first and second RF neighbourhoods of an unrooted binary tree, and we extended the result of Robinson~\cite{Robinson} for the third NNI neighbourhood of an unrooted binary tree (see Appendix A). In addition, we calculated new asymptotic results for the sizes of the $k^{th}$ RF and NNI neighbourhoods of a binary phylogenetic tree. We also found an upper bound on the proportion of binary trees that share at least $k$ non-trivial splits with a given tree on the same leaf set, and found an expression for the number of pairs of binary trees that share a first neighbour under the RF and NNI metrics.\\

In our results for the size of the $k^{th}$ RF and NNI neighbourhoods of an unrooted binary tree $T$ (Theorems~\ref{mainRF} and~\ref{main}), the term of order $n^{k-1}$ contains a parameter dependent on $T$ and~$k$. We have calculated bounds on the value of this parameter: for RF, $-\frac{5k^2+7k}{4}\leq C_{T,k}\leq 4k^2-7k$; for NNI, $\frac{-3k(k+1)}{2}\leq D_{T,k}\leq 3k(k-2)$. These bounds are not strict, so it would be interesting to investigate ways of improving them. A natural question is whether or not both positive and negative values of $C_{T,k}$ and $D_{T,k}$ are possible for any given value of $k$, and if so, whether we can find examples of such trees.\\

We showed that in contrast to RF and NNI, the size of the second SPR neighbourhood is not solely dependent on the number of leaves and cherries of the tree. Humphries and Wu~\cite{TBR} showed that for TBR even the first neighbourhood depends on variables other than the number of leaves and cherries.\\

Throughout this paper, we have considered neighbourhoods of unrooted binary trees under the three metrics; RF, NNI, and SPR. There are, however, many other metrics that can be used to compare trees, and which would be interesting to investigate. For example, Humphries and Wu~\cite{TBR} found an expression for the size of the first TBR neighbourhood of a tree, that depends on variables other than the number of leaves and cherries. Moulton and Wu~\cite{New} recently defined a new metric $d_p$, which is similar to the TBR metric. (The same metric was also independently defined by Kelk and Fischer~\cite{Arxiv}.) Using the result of Humphries and Wu~\cite{TBR}, Moulton and Wu~\cite{New} calculated the size of the first neighbourhood of an unrooted binary tree under this metric. \\

Given the difficulty of calculating the size of the second SPR neighbourhood, it is possible that similar problems would arise in calculating the size of the second neighbourhood under TBR or~$d_p$. However, this would be interesting to investigate, and it may be possible to find the size of the second TBR or $d_p$ neighbourhood of a particular type of tree, such as a caterpillar or a balanced tree. 

\section*{Acknowledgements}

We acknowledge funding support from the University of Canterbury scholarship programme (JdJ), and the Allan Wilson Centre (MS).  We also thank Simone Linz for helpful input in the early stages of this research.

\bibliography{ResearchProposal}
\bibliographystyle{siam}

\begin{appendices}

\section{Third NNI Neighbourhood}

\begin{theorem}
Let $T\in UB(n,c)$ ($n\geq 4$). Then
$$|N^3_{NNI}(T)|=\frac{4}{3}n^3-8n^2-\frac{70}{3}n+8cn+12p_3(T)+164.$$
\label{Third}
\end{theorem}
\begin{proof}
Let $x$ be the number of ways of choosing three distinct internal edges so that no pair is adjacent, let $y$ be the number of ways of choosing these edges so that exactly one pair is adjacent, and let $z$ be the number of ways of choosing these edges so that all pairs are adjacent. Let $t$ be the number of ways of choosing two adjacent edges of~$T$. Robinson~\cite{Robinson} showed that
\begin{equation}|N_{NNI}^3(T)|=8x+16y+24z+36p_3(T)+2t,\label{Rob}\end{equation}
where
\begin{equation}x+y+z+p_3(T)=\frac{(n-3)(n-4)(n-5)}{6},\label{sum}\end{equation}
$$t=n+c-6\leq \frac{3(n-4)}{2},$$
$$p_3(T)\leq 2n-12 \text{ for $n\geq 7$ },$$
$$z=c-2\leq \frac{n-4}{2} \text{ for $n\geq 4$ },$$
and for $n\geq 7$, if $n$ is odd, then $y\leq \frac{3}{2}n^2-16n+42$ and if $n$ is even, then $y\leq \frac{3}{2}n^2-\frac{3}{2}n+45$.\\

There are $(n-5)t$ ways of choosing three distinct internal edges such that at least one pair is adjacent, so we have
$$y=(n-5)(n+c-6)-2p_3(T)-3z,$$
where $2p_3(T)$ is the number of cases where the three edges form a path of length three, and $3z$ is the number of cases where all three edges share an endpoint.\\

It follows from Equation~(\ref{sum})
$$x=\frac{(n-3)(n-4)(n-5)}{6}-y-z-p_3(T).$$
Hence by Equation~(\ref{Rob})
\begin{align*}
|N_{NNI}^3(T)|&=8x+16y+24z+36p_3(T)+2t\\
&=\frac{4}{3}n^3-8n^2-\frac{70}{3}n+8cn+12p_3(T)+164.
\end{align*}
\end{proof}

Theorem~\ref{Third} tells us the size of the third NNI neighbourhood in terms of the number of leaves, cherries and internal paths of length three. We now consider how to determine the number of internal paths of length three.\\

\begin{theorem}
Let $T\in UB(n,c)$ ($n\geq 4$). Then $p_1(T)=n-3$, $p_2(T)=n+c-6$, and for $k\geq 3$,
$$p_k(T)=4p_{k-2}(T)-h_k(T)-m_k(T),$$
where for all values of $k$, $m_k(T)$ is the number of paths of length $k$ in $T$ where both endpoints are leaves of $T$, and $h_k(T)$ is the number of paths of length $k$ in $T$ where exactly one endpoint is a leaf of~$T$.\\
\label{pathsl3}
\end{theorem}
\begin{proof}
The number of internal edges in $T$ is $n-3$, so $p_1(T)=n-3$. The number of pairs of adjacent internal edges is $n+c-6$, so $p_2(T)=n+c-6$. The number of paths of length $k$ in $T$ is $P_k(T)=p_k(T)+m_k(T)+h_k(T).$ Now by Lemma~\ref{ps}, $P_k(T)=4p_{k-2}(T)$. Therefore
$$p_k(T)=P_k(T)-m_k(T)-h_k(T)=4p_{k-2}(T)-h_k(T)-m_k(T).$$
\end{proof}

It follows that $p_3(T)=4(n-3)-h_3(T)-m_3(T)$ for a tree $T\in UB(n,c)$, and therefore
$$ |N^3_{NNI}(T)|=\frac{4}{3}n^3-8n^2+\frac{74}{3}n+8cn-46c-12h_3(T)-12m_3(T)+20.$$
Note that $m_k(T)$ and $h_k(T)$ can both be counted using a breadth--first search in polynomial time. 

\section{Proof of Lemma \ref{two}}

First, suppose that edges $e_1$ and $e_2$ are non-adjacent in~$T$. Let $A$ be the subtree containing $e_2$ such that $d_T(A,e_1)=1$. Let the other three subtrees distance one from $e_1$ be $B$, $C$ and~$D$. First we consider $NNI(T;e_1,e_2)$. The first operation swaps two of the subtrees incident to $e_1$ to obtain $T_1\in NNI(T;e_1)$. We then perform an NNI operation on edge $e_2$ in~$T_1$. We obtain a tree $T_2$ with a subtree $A'$ such that $d_{T_2}(A',e_1)=1$, and $B$, $C$ and $D$ are the other three subtrees distance one from~$e_1$. Now consider $NNI(T;e_2,e_1)$. First, we perform an NNI operation on edge $e_2$ in $A$ (in $T$), and one of the two distinct trees produced is $T_1'\in NNI(T;e_2)$ with subtree $A'$ where $d_{T_1'}(A',e_1)=1$, and $B$, $C$ and $D$ are the other three subtrees distance one from~$e_1$. The second operation swaps two of the subtrees at distance one from $e_1$ in~$T_1'$, which are $A'$, $B$, $C$ and~$D$. One of the two distinct trees obtained is~$T_2$, and so $T_2\in NNI(T;e_2,e_1)$. This is true for all $T_2\in NNI(T;e_1,e_2)$, so~$P\subseteq Q$. Similarly, $Q\subseteq P$ and so~$P=Q$. \\

Now suppose that internal edges $e_1$ and $e_2$ are adjacent. Let $A$ and $B$ be subtrees such that \text{$d_T(A,e_1)=d_T(B,e_1)=1$} and \text{$d_T(A,e_2)=d_T(B,e_2)=2$}. Let $C$ and $D$ be subtrees such that \text{$d_T(C,e_2)=d_T(D,e_2)=1$} and \text{$d_T(C,e_1)=d_T(D,e_1)=2$}. Let $E$ be the subtree such that \text{$d_T(E,e_1)=d_T(E,e_2)=1$}. This can be seen in Fig.~\ref{Extra}. \\

\begin{figure}[H]
\centering
\epsfig{file = 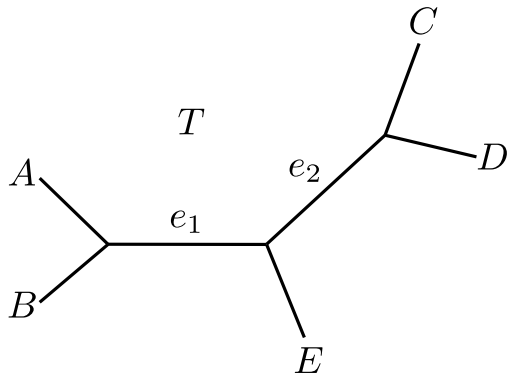, width = 0.5\linewidth}
\caption{A general structure for an unrooted binary tree~$T$, showing subtrees $A$, $B$, $C$, $D$, and~$E$.}
\label{Extra}
\end{figure}
First, we consider $NNI(T;e_1,e_2)$. Let $T_1\in NNI(T;e_1)$ and $T_2\in NNI(T_1;e_2)$. The first operation is on~$e_1$, so either $d_{T_1}(A,E)=2$ or $d_{T_1}(B,E)=2$. Without loss of generality, suppose $d_{T_1}(A,E)=2$. Then $d_{T_1}(E,e_2)=d_{T_1}(A,e_2)=2$. Therefore after the second operation, $d_{T_2}(E,A)=2$. We now consider $NNI(T;e_2,e_1)$. Let $T_1'\in NNI(T;e_2)$ and $T_2'\in NNI(T_1';e_1)$. The first NNI operation is on~$e_2$, so $d_{T_2'}(A,E)=4$. Therefore $d_{T_2'}(A,E)\geq 3$. Hence $T_2'\neq T_2$. The choice of $T_2\in P$ and $T_2'\in Q$ were arbitrary, so $P\cap Q=\emptyset$. \\
\tallqed
\end{appendices}

\end{document}